%% file: lebresneExceptions.tex
\documentclass{LMCS}

\usepackage{amsmath}
\usepackage{amssymb}
\usepackage{stmaryrd}
\usepackage{proof}
\usepackage{framed}
\usepackage{enumerate}
\usepackage{macros}

\theoremstyle{definition}\newtheorem{nota}[thm]{Notation}
\newcommand{\puni}[1]{\suni{#1}{\varepsilon}}
\newcommand{\ppw}[1]{\spw{#1}{\varepsilon}}
\newcommand{\cmd}[1]{\mathtt{#1}}
\newcommand{\labelfont}[1]{\textit{(#1)}}

\def\doi{5 (4:1) 2009}
\lmcsheading%
{\doi}
{1--25}
{}
{}
{Jan.~\phantom{0}9, 2009}
{Nov.~\phantom{0}2, 2009}
{}   

\begin{document}

\title[A Type System For Call-by-name Exceptions]{A Type System For Call-by-name Exceptions}

\author[S.~Lebresne]{Sylvain Lebresne}
\address{Purdue University}
\email{slebresn@purdue.edu}


\begin{abstract}
  We present an extension of System~F with call-by-name exceptions.
  The type system is enriched with two syntactic constructs: a union type for
  programs whose execution may raise an exception at top level, and a
  \emph{corruption type} for programs that may raise an exception in any
  evaluation context (not necessarily at top level).  We present the syntax
  and reduction rules of the system, as well as its typing and subtyping
  rules.  We then study its properties, such as confluence.  Finally, we
  construct a realizability model using orthogonality techniques, from which
  we deduce that well-typed programs are weakly normalizing and that the ones
  who have the type of natural numbers really compute a natural number,
  without raising exceptions.
\end{abstract}

\maketitle

\section{Introduction}
\input{intro.tex}

\section{Design of the system}
\label{presentation}
\input{presentation.tex}

\vfill\eject

\section{Formal presentation}
\label{formal}
\noindent We present the \fx{} calculus, an extension of System F with
typed exceptions, natural numbers and lists.
\subsection{Syntax, reductions and associated properties}
\label{syntax}
\input{syntax.tex}

\subsection{The type system}
\label{typing}
\input{typing.tex}

\section{Examples}
\label{examples}
\input{examples.tex}

\section{Realizability model}
\label{model}
\input{model.tex}

\section{Related Works} 
\label{related}
\input{related.tex}

\section{Conclusion and future works}
\label{futureworks}
\input{futureworks.tex}

\bibliographystyle{plain}
\bibliography{biblio}

\appendix
\section{Parallel reduction for \fx}
\input{appendix.tex}

\end{document}

%% file: intro.tex
\noindent Exceptions are a convenient mechanism for handling errors in 
programming languages. Most modern languages use them: Java, ML,
C++, \ldots.  The main computational features of exceptions are:
\begin{enumerate}[(1)]
  \item One can raise an exception instead of any other expression (or
  instruction);
  \item It propagates automatically by default;
  \item One can catch it only when they need to.
\end{enumerate}

\noindent Exceptions have long been confined to call-by-value
  languages and are usually presented as a mechanism which ``cuts
through'' the normal control flow of a program when raised. This is
viewing the raising of an exception as an effect of the
calculus. Unfortunately, this view makes exceptions hard to transpose
to call-by-name calculi since those do not cope well with effects.

This is a well-known problem. While in call-by-value the effect of a term $t\
u$ can be simply predicted based solely on the effects of $t$ and $u$, in
call-by-name it also depends on the actual term $t$. Indeed, in call-by-name $u$
may well not be evaluated thus not producing its effect (or evaluated many
times, producing the effect many times). For exceptions, this means that in
call-by-name, the fact that $u$ raises an exception does not necessarily imply
that $t\ u$ will. Hence in call-by-name, as summarized by S. Peyton Jones
\textit{et al.}~\cite{jones1999sie}, ``(\ldots) the only productive way to
think about an expression is to consider the \emph{value it computes}, not the
\emph{way in which the value is computed}''. Based on this observation, they
proposed the idea of exceptions-as-values: a value is either a ``normal''
value, or an ``exceptional'' one. In their framework, exceptions are not
effects anymore. And while they present this idea in the context of the
Haskell programming language, this is a very general idea for exceptions in
call-by-name calculi. 

From a typing perspective, exceptions are no simple beasts. Indeed,
the type system should allow the use of exceptions in any part of a program.
What should then be the type of the operation for raising an
exception\footnote{Remark that imperative language alleviate this problem by
making the operation of raising an exception an instruction and not an
expression.}? A solution, used in ML for instance, is to allow the
operation of raising an exception to have any type.
In~\cite{jones1999sie}, S. Peyton Jones \textit{et al.} chose a similar
solution, making exceptional values inhabitants of all types. While simple,
this solution comes at a price, the loss of type safety with respect to
exceptions. The type of an expression never ensures that no exceptions can be
raised during evaluation. 

If we want the type of an expression to reflect which exception it may raise,
a more precise typing is in order. For call-by-value languages, since
exceptions are effects, a convenient and efficient solution is to add to the
type system an effect system~\cite{pessaux1999tba,guzman1994ets}.
Unfortunately and unsurprisingly, this solution is unadapted to the typing of
call-by-name exceptional values. Indeed, tracking values with types is much
more difficult than tracking effects. 

The call-by-name evaluation is well represented amongst type theoretical
calculi which are at the core of many proof assistants (\coq{}~\cite{coq2008},
\textsc{lego}~\cite{lego}, \ldots). We believe this reinforce the case for
studying exceptions in call-by-name and their precise typing. Indeed, the
solution of having exceptional values inhabiting all types would be
inconsistent in these settings. 

This paper proposes a type system for exceptions in call-by-name calculi. By
introducing the new notion of \emph{corruption}, this type system is able to
track which exceptions may escape from a term during evaluation. By using
subtyping, this notion is able to cope with the automatic propagation of
exceptions and to respect the modularity of typing. This type system is
presented in the context of an extension of System~F with exceptions. System~F
is used here as a first step towards more elaborate type theoretical
frameworks.

Meta-theoretical properties of the resulting calculus are proved in
particular by exhibiting a realizability model. While parts of the proofs are
given in this paper, more detailed proofs of the results presented\footnote{A
notable difference is the presence of lists in the language described in this
document, while in~\cite{slphd08} the data type of lists is only presented for
a first-order typed version of the language.} are available in the Ph.D thesis of
the author~\cite{slphd08}. 

The remaining of the paper is organized as follows. We explain our design in
Section~\ref{presentation}: we justify the kind of exception-as-values we use
and describe the three levels of corruption our type system distinguishes.  We
formally present our calculus in Section~\ref{formal} and state the properties
it enjoys and Section~\ref{examples} provides some examples. Then in
Section~\ref{model} we design a realizability model of our calculus that gives
some insight on the meaning of corruption and we prove its soundness. Finally,
we present in Section~\ref{related} some related works before concluding in
Section~\ref{futureworks} with future works.

%% file: presentation.tex
\subsection{Which exceptions-as-values?}
\label{whichexceptions}

As stressed above, exceptions in call-by-name calculi should be values. But
there are essentially two designs for exceptions as values: either we encode
them explicitly in the language, or we make them primitives. The first option
is a well-known one and let us first present its drawbacks in order to justify
the need for the primitive solution. 

Encoding explicitly exceptions is an old
idea~\cite{wadler1985rfl,spivey1990fte}: to each type $A$ is associated a type
$\code{Maybe}\ A$ which is either values of $A$ tagged as correct values or
exceptional values (this idea is nicely explained, for the Haskell programming
language, in~\cite{jones1999sie}). It has later been realized that the $\code{Maybe}$ 
type constructor forms a \emph{monad}~\cite{moggi1991nca,wadler1990cm}. 
And P.~Wadler and P.~Thiemann proposed in~\cite{wadler2003mea} to add effects to
monads, allowing for the detection of uncaught exceptions in such monadic
encoding. However this approach has some drawbacks, namely:
\begin{enumerate}[$\bullet$]
  \item Terms using exceptions are crippled by extra clutter. For example, in
    Haskell, to apply a function \verb!f :: Int -> Int!
    to a value \verb!x :: Maybe Int!
    we are forced to write:
    \begin{verbatim}
    do a <- x
       return (f a)
    \end{verbatim}
    \vspace*{-4mm}
    Using exceptions is not as transparent for the programmer as it is in
    call-by-value languages;
  \item As remarked in~\cite{jones1999sie}, modularity and code re-use are
    compromised, especially for higher order functions. Consider the following
    sorting function:
    \begin{verbatim}
    sort :: (a -> a -> Bool) -> [a] -> [a]
    \end{verbatim}%
    \vspace*{-4mm}
    This function cannot be applied to a comparison function that may raise
    exceptions such as:
    \begin{verbatim}
    cmp :: a -> a -> Maybe Bool
    \end{verbatim}
    \vspace*{-4mm}
    Indeed, with monads, we need to know where the \verb!sort! function uses
    the comparison in order to add the monad's operations;
  \item Monads force the evaluation of arguments (in the example above, the evaluation 
    of \texttt{x} is forced before the application to \texttt{f}). One could not see
    that as an inconvenience, and this is indeed desirable for most uses of
    monads (IO, states, \ldots). Nonetheless, this is a constraint and
    it makes exceptions not usable in non monadic call-by-name code. We think
    that this can be avoided for exceptions. 
\end{enumerate}

\medskip
\noindent This leads us to the second design choice: making exceptions
primitives. This has been first proposed by S.~Peyton~Jones \textit{et
  al.}~\cite{jones1999sie} with \emph{imprecise exceptions}. The idea
is that a value of \emph{any} type is either a ``normal'' value, or an
``exceptional'' one. The resulting mechanism allows exceptions to be
used in place of any other term (as for more traditional
``call-by-value'' exceptions and contrarily to monadic ones). Note
that since values may be exceptional, we can have for instance, a
list, which is fully defined but for which some elements are
exceptional values (see Section~\ref{examples}). These exceptions are
raised only when (and if) the list is evaluated. A main difference
with the call-by-value mechanism of exceptions is for example that a
term like $(\lam{x}{0})\ (\rai)$ (where $0$ is simply the constant
zero and ${\varepsilon}$ some exception) will reduce to $0$ and not to
$\rai$.

Our system, named \fx, adapts this idea to System~F, adding it two new term
constructions: \texttt{raise} and \texttt{try}. But while the exceptions
of~\cite{jones1999sie} are not precisely typed (the raising operation is
in all types), we propose a type system where the type of an expression
indicates which exceptions the expression may raise. 

\input{typingidea.tex}

%% file: typingidea.tex
\subsection{Expected properties}

The type system we will present enjoys the following properties:
\begin{enumerate}[$\bullet$]
  \item If a term can raise an exception, its type indicates it. 
    In particular, programs of type $\nat$ are not able to raise exceptions;
  \item Programmers can use a term $\rai$ in place of any other
    term. In particular, $\rai$ type as a function;
  \item Exceptions and their typing discipline do not jeopardize modularity 
    and code re-use. A function defined without exceptions in mind still
    accepts exceptional arguments and behave in a sensible way. Moreover, this
    is done without knowing the actual code of the function. 
\end{enumerate}

\subsection{Three levels of corruption}
\label{corruption}

We call \emph{corrupted}, a term that may mention exceptions. Given a type $A$ 
(say the type $\nat$ of natural numbers), we distinguish three levels of 
\emph{corruptions} for the terms related with this type:
\begin{enumerate}[$\bullet$]
  \item Terms of $A$. They are not corrupted, either they do not mention 
    exceptions or the ones they mention are caught or erased during reduction; 
  \item Terms of $\puni{A}$. They are terms of $A$ or terms that reduce to
    the exception~${\varepsilon}$, i.e. reduce to $\rai$ (we then say that they raise ${\varepsilon}$). 
  \item Terms of $\ppw{A}$. They are terms of $A$ that may mention the
    exception~${\varepsilon}$ but do not necessarily reduce to it (for instance, if $S$ is
    the successor function, $S\ (\rai)$ has type $\ppw{\nat}$, but not type
    $\puni{\nat}$ since it has not type~$\nat$ nor does it reduce to $\rai$). 
\end{enumerate}
Moreover, to handle the properties of corruption, we use a subtyping relation. 
And in particular we have the subtyping: $A\ \st\ \puni{A}\ \st\ \ppw{A}$. 

The following section explains why the need to distinguish at least those three
levels. But one might wonder why we do not distinguish more levels. Like the
terms containing exceptions but \emph{not} at top level. Or terms having an
exception at a depth of \emph{at most} 2 (like $\rai$ or $\lam{x}{\rai}$ but
not $\lam{x}{\lam{y}{\rai}}$), etc. As of now, while such more precise notion
may well be sound, we have not study them. The main reason is that they would
complicate and clutter the type system while we are not convinced they would
prove useful in practice. 

\subsection{Why we need to distinguish these three levels.}
\label{why3levels}

The construction $\puni{A}$ is really needed because of the typing of the
\texttt{try} operation, since for a \texttt{try} to catch an exception in
its body, this body has to reduce to the exception. 

But because we do not want to change the typing rule of application,
the construction $\puni{A}$ clearly does not fulfill all our
needs. Firstly, we cannot use it to type $S\ (\rai)$.  Secondly, given
a function $M$ of type $A\ {\rightarrow}\ B$, we cannot apply it to a
term $N$ of type $\puni{A}$. Indeed, $M\ N$ is generally not of type
$\puni{B}$ (note however that it would be the case in a call-by-value
calculus). Consider for instance $M = \lam{x}{\lam{y}{x}}$ (of type
$A\ {\rightarrow}\ (C\ {\rightarrow}\ A)$) and $N = \rai$, then $M\ N$
reduces to $\lam{y}{\rai}$ which is not of type
$\puni{(C\ {\rightarrow}\ A)}$ (since it is neither a function of type
$C\ {\rightarrow}\ A$ nor the exception ${\varepsilon}$)\footnote{Note
  that this is typically this example, of the typing of a term like
  $(\lam{x}{\lam{y}{x}})\ N$ when $N$ may raises an exception, that
  makes effect system~\cite{pessaux1999tba} unsuited to call-by-name
  exceptions.}.
%

To solve these problems, we use a second type construction, the
\emph{corruption} of a type $A$ by an exception of name ${\varepsilon}$, denoted 
$\ppw{A}$. The main property the corruption enjoys is a good behavior with
respect to arrow types:
\begin{displaymath}
  \ppw{(\arr{A}{B})}\quad \seq\quad \arr{\ppw{A}}{\ppw{B}}
\end{displaymath}
This subtyping equality\footnote{The subtyping equality $A\ \seq\ B$
is simply defined as shorthand for $A\ \st\ B$ and $B\ \st\ A$. } may seem
paradoxical with the usual subtyping rule of arrow (contra-variance to the
left, co-variance to the right). This is however justified by the
realizability model of Section~\ref{model}.

Intuitively, terms of type $\ppw{A}$ should be seen as terms of type
$A$ where some sub-terms may have been replaced by $\rai$ (hence,
programmers can use $\rai$ wherever they want, which, in turns,
corrupts the resulting type). Equivalently, while terms of $\puni{A}$
are terms that may reduce to $\rai$ at top-level, terms of $\ppw{A}$
are the ones that may reduce to $\rai$ in \emph{any} evaluation
context.

Now, with corruption, we can apply a function $f : \arr{A}{B}$ to a
potentially exceptional term. Indeed, we have that
\begin{displaymath}
  \arr{A}{B}\ \st\ \puni{(\arr{A}{B})}\ \st\ \ppw{(\arr{A}{B})}\ \seq\ \arr{\ppw{A}}{\ppw{B}}.
\end{displaymath}
Remark that since we use subtyping, there is no need to actually know the
term~$f$. This allows for modularity: to type the application of some
(external) function~$f$ to a term~$u$, it is enough to know the type of $f$,
and this even when~$u$ may raise exceptions but the exported type of $f$ does
not mention exceptions.  
This is in particular convenient for primitive functions like the successor
function~$S$, allowing to type-check $S\ (\rai)$ with the type~$\ppw{\nat}$
without the need to give~$S$ a complicated type (the type of $S$ is simply
$\nat\ {\rightarrow}\ \nat$). 

\subsection{Exceptions by the millions}

While we have only used one exception names ${\varepsilon}$ in the above section, it is
useful to be able to handle more than one exception at a time. To that end,
the general type constructions are $\uni{A}$ and $\pw{A}$ where ${\Delta}$ is a set
of exception names. 

Using sets of exceptions requires some type identification using the following
subtyping rules:
\begin{displaymath}
  \begin{array}{c@{\quad \seq\quad }c}
    \uni{(\funi{A}{{\Delta}'})} &  \funi{A}{({\Delta}{\cup}{\Delta}')} \\
    \pw{(\fpw{A}{{\Delta}'})} & \fpw{A}{({\Delta}{\cup}{\Delta}')} \\
    \funi{A}{{\emptyset}} & A \\
    \fpw{A}{{\emptyset}} & A
  \end{array}
\end{displaymath}

%% file: syntax.tex
\subsubsection{Syntax of terms}
\label{syntax:subsyn}

We consider a countable set $\exceptionset$ of names of exceptions and a
distinguished set of variables $\varset$. 
\begin{defi}[Terms]
A \emph{term} of \fx{} is a term generated by the following grammar:
\begin{displaymath}
  \begin{array}{rr@{\quad}l}
    M,N &::=& x \quad|\quad \lam{x}{M} \quad|\quad M\ N \\
    & | &  \rai \quad|\quad \try{M}{N} \\
    & | & 0 \quad|\quad S \quad|\quad \natrec \quad|\quad \nil \quad|\quad \lcons \quad|\quad \foldright
  \end{array}
\end{displaymath}
\end{defi}
In this definition, variables are ranged over by $x, y, \ldots$ while
exception names are ranged over by ${\varepsilon}, {\varepsilon}', \ldots$.
Notions of free and bound variables are defined as usual, as well as the
external operation of substitution (written $\subst{M}{x}{N}$). The set of all
closed terms is denoted $\termset$ and terms are considered up to ${\alpha}$-equivalence.
Note that the construction $\try{M}{N}$ does not bind the occurrences
of ${\varepsilon}$. The term $\rai$ is called an \emph{exception}, ${\varepsilon}$
being its name, but, as an abuse of terminology, we also call ${\varepsilon}$ an
exception. In the term $\try{M}{N}$ we will sometimes call $M$
the body and $N$ the handler of the \texttt{try} construction. 

To the terms of the lambda calculus, we add the constructions to raise and
catch exceptions as well as two usual structured data types: the natural
numbers and the lists. 

\subsubsection{Computation in \texorpdfstring{\fx}{Fx}}
\label{syntax:subcomp}

\begin{defi}[Regular values]
A \emph{regular value} is a (closed) term of \fx{} having one of the following
form:
\begin{displaymath}
  \begin{array}{rr@{\ }l}
    RV &::=& \lam{x}{M} \md 0 \md S \md S\ N \md \natrec 
      \md \natrec\ M \md \natrec\ M\ N  \\
      &|& \nil \md \lcons \md \lcons\ M \md \lcons\ M\ N
      \md \foldright \md \foldright\ M \md \foldright\ M\ N
  \end{array}
\end{displaymath}
\end{defi}

Note that $S\, N$ is a regular value for any term $N$ and hence $S\
(\rai)$ is a regular value as well. 

\begin{defi}[Values]
A \emph{values} is a (closed) term of \fx{} having one of the following
form:
\begin{displaymath}
  V ::=  RV \md \rai
\end{displaymath}
where $RV$ is a regular value and ${\varepsilon}$ any exception name.
\end{defi}

\noindent For well-typed term, a value corresponds to a weak head normal form. 

\begin{defi}[Computation]
The \emph{notion of reduction} $\bred$ for the calculus is defined by the rules of
Figure~\ref{syntax:reduction}. 
Computation in \fx{} is defined from the notion of reduction by the relation
of \emph{reduction}~$\red$ whose rules are given in Figure~\ref{syntax:relreduction}.
We note~$\rede$ the reflexive and transitive closure of~$\red$ and we note
$\eqred$ its reflexive, transitive and symmetric closure. Moreover, if $M
\eqred N$, we will say that $M$ is \emph{equivalent} to $N$. 
\end{defi}
\begin{figure}[htb!]
\begin{framedp}
  \begin{displaymath}
    \begin{array}{r@{\quad\bred\quad}l}
      (\lam{x}{M})\ N & \subst{M}{x}{N} \\
      (\rai)\ M & \rai \vspace*{0.2cm} \\
      \try{(\rai)}{N} & N \\
      \try{(\frai{{\varepsilon}'})}{N} & \frai{{\varepsilon}'} \qquad\rlap{\mbox{(if ${\varepsilon}\ {\neq}\ {\varepsilon}'$)}}\\
      \try{V}{N} & V \qquad\qquad\ \ \;\rlap{\mbox{(if $V$ is a regular value)}}\vspace*{0.2cm} \\
      \natrec\ X\ Y\ 0 & X \\
      \natrec\ X\ Y\ (S\ N) & Y\ N\ (\natrec\ X\ Y\ N) \\
      \natrec\ X\ Y\ (\rai) & \rai \\[.2cm]
      \foldright\ X\ Y\ \nil & X \\
      \foldright\ X\ Y\ (\lcons\ E\ L) & Y\ E\ L\ (\foldright\ X\ Y\ L)\\
      \foldright\ X\ Y\ (\rai) & \rai
    \end{array}
  \end{displaymath}
\end{framedp}
\caption{Notion of reduction for \fx}
\label{syntax:reduction}
\end{figure}
\begin{figure}[htb!]
\begin{framedp}
  \begin{displaymath}
    \begin{array}{c@{\quad}c@{\quad}c}
      \infer{M\red M'}{M\bred M'} &
      \infer{M\ N\red M'\ N}{M \red M'} &
      \infer{\try{M}{N}\red\try{M'}{N}}{M\red M'}\vspace*{1mm}\\
      \infer{\lam{x}{M}\red\lam{x}{M'}}{M \red M'} &
      \infer{M\ N\red M\ N'}{N \red N'}&
      \infer{\try{M}{N}\red\try{M}{N'}}{N\red N'}
    \end{array}
  \end{displaymath}
\end{framedp}
  \caption{Relation of reduction for \fx}
  \label{syntax:relreduction}
\end{figure}

Note that, as usual, the scope of capture of the \texttt{try}
construction is dynamic: in the term $(\lam{x}{\try{x}{0}})\ (\rai),$
the exception is caught during reduction and 
the whole term reduces to~$0$. 
We say that a term $M$ \emph{raises the exception} ${\varepsilon}$ if
$M\rede\rai$ (that is, if $M$ reduces to the exception named ${\varepsilon}$).

\begin{defi}[to have a value]
  We will say that a term $M$ \emph{has a value} if and only if it reduces to
  a value, that is if there exists a value $V$ such that $M\rede V$.
\end{defi}
It can be proved~\cite{slphd08} that this notion is equivalent to the one of
having a normal form for the weak head reduction of the calculus. 

We now show that adding \texttt{raise} and \texttt{try} does not break the
confluence of the calculus:
\begin{thm}[Confluence]
  \label{syntax:confluence}
  If $M$, $N$ and $N'$
  are terms such that $M\rede N$ and $M\rede N'$, then there exists a term $P$
  such that $N\rede P$ and $N'\rede P.$
\end{thm}
\begin{proof}
We adapt the proof originated by Tait and Martin-L\"of for
the confluence of pure lambda-calculus that can be found
in~\cite{barendregt1984lc} for example. We define the notion of parallel
reduction $\pred$ for \fx, we show that it satisfies the diamond property and conclude
since $\rede\ =\ \prede$. Proofs of these properties are easy to tackle
inductions we leave to the interested reader. We however give in
appendix~\ref{appendix:pred} the definition of the parallel reduction for~\fx{}. 
\end{proof}

%
%
%

%% file: typing.tex
As stressed in Section~\ref{corruption}, \fx{} uses a subtyping relation
$\st$. Thus, \fx{} is in fact an extension of the second-order lambda calculus
with subtyping introduced by Mitchell~\cite{mitchell1988pti,wells1995ums} (and
we will call this calculus System~F${\eta}$ in the following). Note that we will
however use a presentation of this calculus that differs from the original one
and that can be found for example in~\cite{tiuryn1996sps}. 

\begin{defi}[Types]
The syntax of types for \fx{} is built upon the one of System F. \emph{Type}
of \fx{} are generated by the following grammar:
\begin{displaymath}
  A,B ::= {\alpha} \md \nat \md \liste{A} \md \arr{A}{B} \md \fall{{\alpha}}{A} \md \uni{A} \md \pw{A} 
\end{displaymath}
\end{defi}
In $\uni{A}$ and $\pw{A}$, ${\Delta}$ is a finite set of exceptions
names (${\Delta}\ {\subseteq}\ \exceptionset$). Moreover, ${\alpha}$ stands for a type variable
taken from the set of type variables $\typevarset$. Notions of free and bound
type variable are defined as usual, as well as the external operation of
substitution (written $\subst{A}{{\alpha}}{B}$). We denote by $FV(A)$ the set of all
the free type variables of the type $A$. Types are considered up to
${\alpha}$-equivalence.  Precedences for the arrow construction and the universal
quantifier are the usual ones; the precedences of
$\uni{A}$ and $\pw{A}$ being higher. Moreover, we will often write
$\pw{\liste{A}}$ for $\pw{(\liste{A})}$.

\subsubsection{Typing}
\begin{defi}[Typing context]
A \emph{typing context} ${\Gamma}$ is a finite set of declarations having the form
${\Gamma}\ {\equiv}\ x_1 : A_1, \ldots, x_n : A_n$ where $x_1,\ldots,x_n$ are pairwise
distinct term variables and where $A_1,\ldots, A_n$ are arbitrary types.
\end{defi}
The set $FV({\Gamma})$ of free variables of ${\Gamma}$ denotes the union of the sets of free type variables for the
types used in ${\Gamma}$, that is to say: 
\begin{displaymath}
  FV(x_1 : A_1, \ldots, x_n : A_n)\ = \displaystyle\mathop{{\bigcup}}_{i\,{\in}\;\lens{1\ldots n}}FV(A_i)
\end{displaymath}

\begin{defi}[Typing]
The type system of \fx{} is defined from the \emph{typing judgment}
$$\jtyp{M}{A}$$ 
that reads `in the typing context ${\Gamma}$, the term $M$ has type
$A$'. This judgment is inductively defined by the rules of
Figure~\ref{typing:typing}. 
\end{defi}
Remark that the typing rules from System~F${\eta}$ are unchanged, we simply add
rules. Also note that the usual typing rules for the recursion operators can
be retrieved from \textit{(rec)} and \textit{(fold)} by taking ${\Delta}={\emptyset}$ (theses
rules are in fact typing schemes).
\newcommand{\esp}{\hspace*{1cm}}
\begin{figurep}
\begin{framedp}
  \begin{center}
  \textbf{System~F${\eta}$ typing rules:}
  \end{center}
  \vspace*{2mm}
  \begin{displaymath}
    \begin{array}{c@{\qquad}c@{\qquad}c}
      \inferpb{\labelfont{ax}}{\jtyp{x}{A}}{(x:A)\,{\in}\, {\Gamma}} 
      &
      \inferpb{\labelfont{abs}}{\jtyp{\lam{x}{M}}{\arr{A}{B}}}{\jtypctx{{\Gamma}, x:A}{M}{B}}
      &
      \inferpb{\labelfont{app}}{\jtyp{M\ N}{B}}{\jtyp{M}{\arr{A}{B}} & \jtyp{N}{A}}
    \end{array}
  \end{displaymath}
  \begin{displaymath}
    \begin{array}{c@{\qquad\qquad}c}
      \inferpb{\labelfont{gen}}{\jtyp{M}{\fall{{\alpha}}{A}}}{\jtyp{M}{A} & {\alpha}\,{\notin}\;FV({\Gamma})}
      &
      \inferpb{\labelfont{subs}}{\jtyp{M}{B}}{\jtyp{M}{A} & A\;\st\;B}
    \end{array}
  \vspace*{4mm}
  \end{displaymath}
  \begin{center}
  \textbf{Natural numbers typing rules:}
  \end{center}
  \vspace*{2mm}
  \begin{displaymath}
    \begin{array}{c@{\esp\esp}c}
      \inferpbb{\labelfont{zero}}{\jtyp{0}{\nat}}{}
      &
      \inferpbb{\labelfont{succ}}{\jtyp{S}{\arr{\nat}{\nat}}}{}
    \end{array}
  \end{displaymath}
  \vspace*{-2mm}
  \begin{displaymath}
      \inferpbb{\labelfont{rec}}{\jtyp{\natrec}{\fall{{\alpha}}{\uni{{\alpha}}\;{\rightarrow}\;(\pw{\nat}\;{\rightarrow}\;\uni{{\alpha}}\;{\rightarrow}\;\uni{{\alpha}})\;{\rightarrow}\;\funi{\pw{\nat}}{{\Delta}'}\;{\rightarrow}\;\funi{{\alpha}}{({\Delta}{\cup}{\Delta}')}}}}{}\vspace*{-1mm}
  \vspace*{4mm}
  \end{displaymath}
  \begin{center}
  \textbf{List typing rules:}
  \end{center}
  \vspace*{2mm}
  \begin{displaymath}
    \begin{array}{c@{\esp\esp}c}
      \inferpbb{\labelfont{nil}}{\jtyp{\nil}{\fall{{\alpha}}{\liste{{\alpha}}}}}{}
      &
      \inferpbb{\labelfont{cons}}{\jtyp{\lcons}{\fall{{\alpha}}{{\alpha}\;{\rightarrow}\;\liste{{\alpha}}\;{\rightarrow}\;\liste{{\alpha}}}}}{}
    \end{array}
  \end{displaymath}
  \vspace*{-2mm}
  \begin{displaymath}
      \inferpbb{\labelfont{fold}}{
       \begin{array}{l}
        \jtyp{\foldright}{\fall{{\alpha}}{\fall{{\beta}}{\uni{{\alpha}}\;{\rightarrow}\;(\pw{{\beta}}\;{\rightarrow}\;\pw{\liste{{\beta}}}\;{\rightarrow}\;\uni{{\alpha}}\;{\rightarrow}\;\uni{{\alpha}})}}}\\
        \qquad\qquad\qquad\qquad\;{\rightarrow}\;\funi{\pw{\liste{{\beta}}}}{{\Delta}'}\;{\rightarrow}\;\funi{{\alpha}}{({\Delta}{\cup}{\Delta}')}
       \end{array}
        }{}\vspace*{-1mm}
  \vspace*{4mm}
  \end{displaymath}
  \begin{center}
  \textbf{Exceptions handling typing rules:}
  \end{center}
  \vspace*{2mm}
  \begin{displaymath}
    \begin{array}{c@{\esp\esp}c}
      \inferpb{\labelfont{raise}}{\jtyp{\rai}{\fall{{\alpha}}{\funi{{\alpha}}{\lens{{\varepsilon}}}}}}{}
      &
      \inferpb{\labelfont{try}}{\jtyp{\try{M}{N}}{A}}{\jtyp{M}{\funi{A}{\lens{{\varepsilon}}}} & \jtyp{N}{A}}
    \end{array}
  \end{displaymath}
\end{framedp}
\caption{Typing judgments}
\label{typing:typing}
\end{figurep}

\subsubsection{Subtyping}
\begin{defi}[Subtyping]
The \emph{subtyping} relation between two types $A$ and $B$, written \mbox{$A \st
B$}, is inductively defined by the rules of Figure~\ref{typing:subtyping}.
\end{defi}
The equality $A \seq B$ is defined as short for ``$A \st B$ and $A \ist B$''. 
In the inference rules, when the equality $A \seq B$ appears as a premise, it
figures for the two premises $A \st B$ and $A \ist B$. And when it appears as
a conclusion, it figures for two inference rules, one having $A \st B$ as a
conclusion, the other one having $A \ist B$.
\renewcommand{\esp}{\hspace*{1cm}}
\begin{figurep}
\begin{framedp}
  \textbf{System F${\eta}$ rules~:}
  \vspace*{1mm}

  \begin{displaymath}
    \begin{array}{c@{\quad}c@{\quad}c}
      \inferp{\labelfont{st-id}}{\jst{A}{A}}{} 
      &
      \inferp{\labelfont{st-trans}}{\jst{A}{C}}{\jst{A}{B} & \jst{B}{C}}
      &
      \inferp{\labelfont{st-arrow}}{\jst{\arr{A}{B}}{\arr{A'}{B'}}}{\jst{A'}{A} & \jst{B}{B'}}
    \end{array}
  \end{displaymath}
  \begin{displaymath}
    \begin{array}{c@{\quad\qquad}c}
      \inferp{\labelfont{f-gen}}{\jst{A}{\fall{{\alpha}}{B}}}{\jst{A}{B} & {\alpha}\,{\notin}\;FV(A)}
      &
      \inferp{\labelfont{f-inst}}{\jst{\fall{{\alpha}}{A}}{\subst{A}{{\alpha}}{B}}}{}
    \end{array}
  \end{displaymath}
  \begin{displaymath}
      \inferp{\labelfont{f-distr}}{\jst{\fall{{\alpha}}{(\arr{A}{B})}}{\arr{A}{\fall{{\alpha}}{B}}}}{{\alpha}\,{\notin}\;FV(A)}
  \vspace*{2mm}
  \end{displaymath}
  \textbf{Exception related rules~:}
  \vspace*{1mm}

  \begin{displaymath}
    \begin{array}{c@{\qquad\qquad}c}
      \inferp{\labelfont{ex-uni}}{\jst{A}{\uni{A}}}{}
      &
      \inferp{\labelfont{ex-corrupt}}{\jst{\uni{A}}{\pw{A}}}{}
    \end{array}
  \end{displaymath}
  \begin{displaymath}
    \inferp{\labelfont{ex-noexc}}{\jst{\fpw{A}{\emptyset}}{A}}{}
  \end{displaymath}
  \begin{displaymath}
    \begin{array}{c@{\qquad\qquad}c}
        \inferp{\labelfont{ex-ctx}}{\jst{\uni{A}}{\uni{B}}}{\jst{A}{B}}
        &
        \inferp{\labelfont{ex-arru}}{\jst{\uni{(\arr{A}{B})}}{\arr{A}{\uni{B}}}}{} 
    \end{array}
  \end{displaymath}
  \vspace*{.2cm}
  \begin{displaymath}
    \begin{array}{c@{\esp}c}
      \inferp{\labelfont{ex-fallc}}{\jst{\fall{{\alpha}}{\pw{A}}}{\pw{(\fall{{\alpha}}{A})}}}{} 
        &
      \inferp{\labelfont{ex-fallu}}{\jst{\fall{{\alpha}}{(\uni{A})}}{\uni{(\fall{{\alpha}}{A})}}}{} 
     \end{array}
  \vspace*{2mm}
  \end{displaymath}
  \begin{displaymath}
    \begin{array}{c@{\esp\quad}c}
      \inferp{\labelfont{ex-lcor}}{\jst{\liste{\pw{A}}}{\pw{\liste{A}}}}{} 
        &
      \inferp{\labelfont{ex-lctx}}{\jst{\liste{A}}{\liste{B}}}{\jst{A}{B}} 
    \end{array}
  \vspace*{2mm}
  \end{displaymath}

  \textbf{Exception related equality rules~:}
  \vspace*{1mm}

  \begin{displaymath}
      \begin{array}{c@{\esp}c}
        \inferp{\labelfont{eq-uu}}{\funi{(\uni{A})}{{\Delta}'} \ \seq\  \funi{A}{({\Delta}{\cup}{\Delta}')}}{} 
        &
        \inferp{\labelfont{eq-cc}}{\fpw{(\pw{A})}{{\Delta}'}\ \seq\ \fpw{A}{({\Delta}{\cup}{\Delta}')}}{} 
      \end{array}
  \end{displaymath}
  \vspace*{.1cm}
   \begin{displaymath}
     \begin{array}{c@{\esp}c}
        \inferp{\labelfont{eq-uc}}{\fpw{(\uni{A})}{{\Delta}'}\ \seq\ \uni{\fpw{A}{{\Delta}'}}}{}
        &
        \inferp{\labelfont{eq-arrc}}{\pw{(\arr{A}{B})}\ \seq\ \arr{\pw{A}}{\pw{B}}}{}
      \end{array}
  \end{displaymath}
\end{framedp}
\caption{The subtyping relation}
\label{typing:subtyping}
\end{figurep}

The subtyping rules from F${\eta}$ are unchanged. The rules
\textit{(ex-noexc)}, \textit{(eq-uu)} and \textit{(eq-cc)} dealt with sets of
exceptions. The hierarchy of corruption (see~\ref{corruption}) is implemented
by \textit{(ex-uni)} and \textit{(ex-corrupt)}. The rules \textit{(ex-fallc)} and
\textit{(ex-fallu)} are justified by the absence of computational content of
the universal quantification. Moreover, corruption and union commutes 
\textit{(eq-uc)}.

The subtyping is stable by union \textit{(ex-ctx)}, but also by corruption 
(this is proved by Theorem~\ref{theorem:corruptstable}).
Rule \textit{(ex-arru)} simply says that, since a term~$M$ of
type~$\uni{(\arr{A}{B})}$ is either a term of type~$\arr{A}{B}$ or an
exception of~${\Delta}$, it can always be applied to a term of type~$A$, resulting in
a term of type $B$ (if $M$ is a true function) or an exception of~${\Delta}$ (if so
is~$M$).

As discussed in Section~\ref{why3levels}, the rule \textit{(eq-arrc)}
is the main rule of corruption and allows exceptions to be used anywhere. Note
that we really need an equality here on pain of losing the subject-reduction
property.

Finally, the list construction is monotonic (rule \textit{(ex-lcor)}) and a
list of corrupted elements is in particular a corrupted list (rule
\textit{ex-lctx}).

\medskip
The subtyping associated to the notion of corruption is a quite flexible one,
especially with respect to arrows. As noted in Section~\ref{why3levels}, it
allows to derive 
\begin{displaymath}
  \jst{A_1\ \to\ A_2\ \to\ \ldots\ \to\ A_n}{\pw{A_1}\ \to\ \pw{A_2}\ \to\ \ldots\ \to\ \pw{A_n}}
\end{displaymath}
but also that $\jst{\arr{\pw{A}}{B}}{\arr{\pw{A}}{\pw{B}}}$ or that
$\jst{\arr{A}{\pw{B}}}{\arr{\pw{A}}{\pw{B}}}$ for instance\footnote{Proofs for
all those relations follow the same pattern. Corruption is introduced on the
right with \textit{(ex-uni)} and \textit{(ex-corrupt)} and is then
distributed over the operands of the arrow with \textit{(eq-arrc)}. Lastly,
double corruption ($\pw{(\pw{A})}$) is eliminated with \textit{(eq-cc)} if
needed.}. However, what the subtyping of corruption forbids is the removable
of corruption in covariant position. That is, corruption allows the use of
functions with exceptions they do not handle themselves, but it then always
ensure that the return type mentions those exceptions. 

\subsubsection{Typing the recursion operations}
\label{typing:tailrec}

\fx{} uses natural numbers and lists. To work with these data types, we
have equipped the calculus with recursion operators ($\natrec$ and
$\foldright$) of Gödel's System T\cite{girard1989pat}. However, in the presence of
exceptions, the usual typing of these operators is not precise enough. Indeed,
consider the case of the natural numbers. The usual typing rule of the
recursion operator
$\natrec$ is $$\fall{{\alpha}}{{\alpha}\ {\rightarrow}\ (\nat\ {\rightarrow}\ {\alpha}\ {\rightarrow}\ {\alpha})\ {\rightarrow}\ {\alpha}\ {\rightarrow}\ {\alpha}},$$
and hence, using the corruption type and its associated subtyping rules, it
also has the type
$$\fall{{\alpha}}{\pw{{\alpha}}\ {\rightarrow}\ (\pw{\nat}\ {\rightarrow}\ \pw{{\alpha}}\ {\rightarrow}\ \pw{{\alpha}})\ {\rightarrow}\ \pw{{\alpha}}\ {\rightarrow}\ \pw{{\alpha}}}$$
for any set of exception names ${\Delta}$. However, this last type is not precise enough
and for instance, it does not precisely account for the reduction rule 
$$\natrec\ X\ Y\ (\rai)\ \red\ \rai.$$
Dealing with this imprecision is the reason of the addition of the set ${\Delta}'$ in
the typing rule of $\natrec$ (rule \labelfont{rec} of
Figure~\ref{typing:typing}). Moreover, the function $eval$ which will be
introduced in section~\ref{examples} reveals another imprecision. Given a
corrupted natural number, this function returns either a well formed natural
number or an exception at top-level. But to give this function the type we
want, that is to say the type $\arr{\pw{\nat}}{\uni{\nat}}$, we need the
addition of the set ${\Delta}$ in the typing rule of the recursion operator (rule
\labelfont{rec} of Figure~\ref{typing:typing}). The typing rule of the
recursion operator $\foldright$ follows the same modifications. 

%

\subsection{Properties of typing}

The subtyping relation is stable by corruption:
\begin{thm}
  \label{theorem:corruptstable}
  If $A$ and $B$ are two types such that $A \st B$, then for any set of
  exception names ${\Delta}$, $\pw{A} \st \pw{B}$.
\end{thm}
\proof
  We proceed by induction on the derivation of $A \st B$. All the cases are
  easily resolved since corruption commutes with all type constructions. For
  example, taking the case of rule \tlabel{ex-arru}, we have to show that
  $\pw{(\funi{(\arr{A}{B})}{{\Delta}'})}\ \st\ \pw{(\arr{A}{\funi{B}{{\Delta}'}})}.$
  But using rule \tlabel{ex-arru}, 
  $\funi{(\arr{\pw{A}}{\pw{B}})}{{\Delta}'}\ \st\ \arr{\pw{A}}{\funi{\pw{B}}{{\Delta}'}}$
  and we conclude using the fact that 
    $\funi{(\arr{\pw{A}}{\pw{B}})}{{\Delta}'}\ \seq\ \pw{(\funi{(\arr{A}{B})}{{\Delta}'})}$
  and
    $\arr{\pw{A}}{\funi{\pw{B}}{{\Delta}'}}\ \seq\ \pw{(\arr{A}{\funi{B}{{\Delta}'}})}$.\qed

\noindent A few remarkable subtyping rules are also easily derivable from the ones
of Figure~\ref{typing:subtyping}:
\begin{thm}
  The following subtyping relations hold:
  \begin{displaymath}
    \begin{array}{rcl}
      \funi{A}{{\emptyset}} & \st & A \\
      \uni{(\fall{{\alpha}}{A})} & \st & \fall{{\alpha}}{(\uni{A})} \\
      \pw{(\fall{{\alpha}}{A})} & \st & \fall{{\alpha}}{(\pw{A})} 
    \end{array}
  \end{displaymath}
\end{thm}
\proof
  Proof of $\funi{A}{{\emptyset}}\st A$ comes from rules~\tlabel{ex-corrupt}
  and~\tlabel{ex-noexc}. The proofs for $\uni{(\fall{{\alpha}}{A})} \st
  \fall{{\alpha}}{(\uni{A})}$ and $\pw{(\fall{{\alpha}}{A})} \st \fall{{\alpha}}{(\pw{A})}$ are
  similar. For instance, for the former one, we use \tlabel{f-inst} and
  \tlabel{ex-uni} to show that $\fall{{\alpha}}{A}\st \uni{A}$. Then, using
  \tlabel{ex-ctx} and \tlabel{eq-uu}, we show that $\uni{(\fall{{\alpha}}{A})}\st
  \uni{A}$. And we conclude with \tlabel{f-gen}.\qed

\begin{defi}
The \emph{corruption relation} $\morecor{{\Delta}}{}{}$ between terms is inductively
defined Figure~\ref{fig:typing:morecorruption}. To have $\morecor{{\Delta}}{M}{N}$
means that $N$ is obtained from $M$ by replacing some sub-terms in any
position by $\rai$, ${\varepsilon}$~belonging to ${\Delta}$.
\end{defi}
\begin{figurep}
  \begin{framedp}
    \begin{displaymath}
      \begin{array}{c@{\qquad}c@{\qquad}c}
        \infer[\tlabel{c-id}]{\morecor{{\Delta}}{M}{M}}{} &
        \infer[\tlabel{c-rai}]{\morecor{{\Delta}}{M}{\rai}}{{\varepsilon}\,{\in}\,{\Delta}} &
        \infer[\tlabel{c-lam}]{\morecor{{\Delta}}{\lam{x}{M}}{\lam{x}{t'}}}{\morecor{{\Delta}}{M}{M'}}
      \end{array}
    \end{displaymath}
    \begin{displaymath}
      \begin{array}{c@{\quad}c}
        \infer[\tlabel{c-app}]{\morecor{{\Delta}}{M\ N}{M'\ N'}}{\morecor{{\Delta}}{M}{M'} & \morecor{{\Delta}}{N}{N'}}
        & 
        \infer[\tlabel{c-try}]{\morecor{{\Delta}}{\try{M}{N}}{\try{M'}{N'}}}{\morecor{{\Delta}}{M}{M'} & \morecor{{\Delta}}{N}{N'}}
      \end{array}
    \end{displaymath}
  \end{framedp}
  \caption{Corruption relation}
  \label{fig:typing:morecorruption}
\end{figurep}

Thus, Theorem~\ref{typing:thmcorruption} formally states that, in term of
programming, exceptions can be used in any place, but with the added cost of
corrupting the type.
\begin{thm}[corruption]
  \label{typing:thmcorruption}
    If $M$ and $N$ are two terms, $A$ a type and ${\Delta}$ a set of exceptions such 
    that $\jtyp{M}{A}$ and $\morecor{{\Delta}}{M}{N}$, 
    then $\jtyp{N}{\pw{A}}$.
\end{thm}
\proof
This theorem is proved by induction on the statement $\morecor{{\Delta}}{M}{N}$. The proof
presents no major difficulty as long as we first prove the three following
``inversion'' results~:
\begin{enumerate}[(1)]
  \item If $M$ is a term, $A$ a type and ${\Gamma}$ a typing context such that 
    $$\jtyp{\lam{x}{M}}{A},$$ then there exists a set of type variable
    $\overrightarrow{{\alpha}}\,{\notin}\;FV({\Gamma})$ and two terms $B$ and $C$ such that
    $\jst{\fall{\overrightarrow{{\alpha}}}{(\arr{B}{C})}}{A}$ and $\jtypctx{{\Gamma},
    x:B}{M}{C}$.
  \item If $M$ and $N$ are two terms, $A$ a type and ${\Gamma}$ a typing context such
    that $$\jtyp{M\ N}{A}\ ,$$ then there exists a term $C$ such that
    $\jtyp{M}{\arr{C}{A}}$ and $\jtyp{N}{C}$.
  \item If $M$ and $N$ are two terms, $A$ is a type and ${\Gamma}$ is a typing context such
    that $$\jtyp{\try{M}{N}}{A},$$ then 
    $\jtyp{M}{\suni{A}{{\varepsilon}}}$ and $\jtyp{N}{A}$.
\end{enumerate}
Proofs of these three results are straightforward inductions on the derivation
of the initial typing judgment.\qed

%% file: examples.tex
\noindent A simple yet classical function on natural numbers which can
raise an exception is the predecessor function. In \fx{}, we can define:
\begin{displaymath}
  \cmd{pred}\quad {\equiv}\quad \natrec\ (\rai)\ (\lam{x}{\lam{y}{x}})\quad : \quad
  \arr{\nat}{\suni{\nat}{\exn{pred\_err}}}
\end{displaymath}
It has the expected reductions, i.e. $\cmd{pred}\ 0 \rede\rai$ and $\cmd{pred}\ (S\
N)\rede N$. We can then define a ``safe'' predecessor $\cmd{pred'}$ from
$\cmd{pred}$ which returns $0$ when applied to $0$:
\begin{displaymath}
  \cmd{pred'}\quad {\equiv}\quad \lam{n}{\ftry{(\cmd{pred}\ n)}{\exn{pred\_err}}{0}} \quad : \quad \arr{\nat}{\nat}
\end{displaymath}

Having exceptions, it is possible to define the functions that return the
head and the tail of a list:
\begin{displaymath}
  \begin{array}{rcl}
    hd &:& \arr{\liste{A}}{\funi{A}{\lens{\exn{hd\_fail}}}} \\
       &{\equiv}& \foldright\ (\frai{\exn{hd\_fail}})\ (\lam{e}{\lam{l}{\lam{\_}{e}}})\\[0.2cm]
    tl &:& \arr{\liste{A}}{\funi{\liste{A}}{\lens{\exn{tl\_fail}}}} \\
       &{\equiv}& \foldright\ (\frai{\exn{tl\_fail}})\ (\lam{e}{\lam{l}{\lam{\_}{l}}})
  \end{array}
\end{displaymath}

We can also define the Euclidean division ($div : \nat\ {\rightarrow}\ \nat\ {\rightarrow}\
\funi{\nat}{\lens{\exn{div\_by\_0}}}$) and the mapping  of a function to
a list ($map : \fall{{\alpha}}{\fall{{\beta}}{({\alpha}\ {\rightarrow}\ {\beta})\ {\rightarrow}\ \liste{{\alpha}}\ {\rightarrow}\ \liste{{\beta}}}}$).
The type system being modular, we can (and will) use these two functions
without having to exhibit a particular implementation. 
%
But using them allows us to define the following function that maps the function
$n\ {\mapsto}\ \frac{2}{n-1}$ to a list:
\begin{displaymath}
  \begin{array}{rcl}
    f &:& \liste{\nat}\ {\rightarrow}\ \liste{(\funi{\fpw{\nat}{\lens{\exn{pred\_err}}}}{\lens{\exn{div\_by\_0}}})}\\
        &{\equiv}& map\ (\lam{n}{div\ 10\ (pred\ n)}) 
  \end{array} 
\end{displaymath}
Remark that the result is always a list but that can contain exceptional
values. For instance, $f\ [2; 1; 5]$ computes the list $[5;
\frai{\exn{div\_by\_0}}; 2]$ which does not reduces to
$\frai{\exn{div\_by\_0}}$. Again, exceptions are values that propagate only
when used. 
Now we can get the first element of the result of this function with:
\begin{displaymath}
  \begin{array}{rcl}
    g &:& \liste{\nat}\ {\rightarrow}\ \funi{\fpw{\nat}{\lens{\exn{pred\_err},
    \exn{div\_by\_0}}}}{\lens{\exn{hd\_fail}}}\\
      &{\equiv}& \lam{l}{hd\ (f\ l)}
  \end{array} 
\end{displaymath}
We can apply $g$ to some argument and catch the exception $\exn{hd\_fail}$
with a $\code{try}$, but we cannot catch the two other exceptions since these
ones are not necessarily at top-level. If we want to catch them, we need a
function that evaluates a natural number potentially corrupted and returns
either a well formed natural number or an exception. It is the purpose of the
following function:
\begin{displaymath}
  \begin{array}{rcl}
    eval &:& \arr{\pw{\nat}}{\uni{\nat}}\\
         &{\equiv}& \lam{n}{(\natrec\ (\lam{a}{a})\ (\lam{m}{\lam{r}{\lam{a}{r\ (S\ a)}}})\ n)\ 0}
  \end{array}
\end{displaymath}
With this function, we can now capture the exceptions that can appear in
the result of the function $g$ above. That is what the following function does
(where we use a straightforward shortcut allowing the $\code{try}$ to catch
all the exceptions):
\begin{displaymath}
  \begin{array}{rcl}
    h &:& \arr{\liste{\nat}}{\nat} \\
      &{\equiv}& \lam{l}{\mathtt{try}}\ (eval\ (g\ l))\\
      & & \hphantom{\lam{l}{}}\mathtt{with}\:\exn{pred\_err},\exn{div\_by\_0},\exn{hd\_fail}~{\mapsto}~0
  \end{array}
\end{displaymath}
Note that for instance $h\ [2; 1; 5]$ will return $5$ in our system since the
part of the list that would yield an exception (the second element after the
mapping) is never used (we only use the head of the list). In contrast, a
similar function in say Caml would have yielded~$0$.

%% file: model.tex
\noindent We will define a realizability model for \fx{} using techniques of
orthogonality (see~\cite{parigot1993sns,vouillon2004stf} for examples of use
of such techniques). The choice of those orthogonality techniques is mainly
motivated by two reasons: we believe that it offers a nice way to handle
second order and it will come in handy for the definition of the
interpretation of corruption, allowing a much more simple definition than a
``direct'' model would allow. We start by introducing a few definitions
necessary to the construction of the model.

\subsection{Daimon and contexts}

We add a new and distinguished term, the \emph{daimon} (denoted~$\daemon$)
similar to the one of \cite{girard2001lsr}. This term computationally behaves
like an uncatchable exception. We also introduce the new term construction $M;
N$. This construction tests if $M$ is the daimon and if so, return~$N$.
Otherwise, it does not reduce. The reduction rules for these two additions are
given Figure~\ref{fig:model:daimon}. Moreover, $\daemon$ is added to the
definition of value. 
Note that none of these constructs have typing rules and as such, they
cannot be used in well-typed terms. It can also easily be proved that they do
not break the confluence property of the language. In those respects, they are
only convenient technical addition for the model and should not be considered
as inherent part of the language.  
\begin{figurep}
\begin{framedp} 
  \begin{displaymath}
    \begin{array}{r@{\quad\bred\quad}l}
      \daemon\ N & \daemon \\
      \try{\daemon}{N} & \daemon \\
      \natrec\ X\ Y\ \daemon & \daemon \\
      \foldright\ X\ Y\ \daemon & \daemon \\[1mm]
      \daemon; N & N
    \end{array}
  \end{displaymath}
\end{framedp}
\caption{Reduction rules for $\daemon$ and $;$. }
\label{fig:model:daimon}
\end{figurep}

The daimon has two purposes in the model. First, it will inhabit all type
interpretation, property that will be used to show that all the terms of the
interpretation are weakly normalizing (see
Lemma~\ref{lemme:fx:model:modelwhnf}). Secondly, our model is a realizability
one, types will be interpreted by sets of terms. But the principle of our
orthogonality model is to not define those sets directly, but instead to
first define the interpretation of types as sets of evaluation contexts. Then,
to each such set $S$ of evaluation contexts is associated the set of all the
terms that ``behave correctly'' for all the contexts of $S$. This notion of
a term $M$ ``behaving correctly'' in a context $C$ is the orthogonality
relation. For our model we chose it to be that $M$ put in the context $C$
reduces to this distinguished term $\daemon$. But to define formally this
orthogonality relation, let us first define formally the evaluation contexts
we will use: 
\begin{defi}[Context]
A \emph{context} is a term with a hole (denoted by $\hole$) defined by:
\begin{displaymath}
  C ::= \hole \md C\ N \md \try{C}{\daemon} \md \natrec\ M\ N\ C \md
  \foldright\ X\ Y\ C
\end{displaymath}
The set of all contexts is noted $\contextset$ and the term obtained by
filling the hole of a context $C$ with the term $M$ is noted $C[M]$. 
\end{defi}
Note that our definition of context is more restrictive than the usual
one (where a context is \emph{any} term with a hole). Actually, save the
restriction in the handler of \texttt{try} to $\daemon$ (which will allow for
a simpler interpretation of corruption), our contexts are the evaluation
contexts of call-by-name evaluation. 

Moreover, we will not care about the order of two adjacent \texttt{try} in a
context. Since the set of all exception names $\exceptionset$ is countable, we
can fix \textit{a priori} a bijection ${\phi} : \exceptionset\;\to\;\mathcal{N}$
and we define the following notation:
\begin{nota}
  If ${\Delta}$ is a (finite) set of exception names, then 
    $\stry{\hole}$
  is a notation for
  \begin{enumerate}[$\bullet$]
    \item the context
      $\ftry{(\ldots(\ftry{\hole}{{\varepsilon}_1}{\daemon})\ldots)}{{\varepsilon}_n}{\daemon}$
      if ${\Delta}\ {\neq}\ {\emptyset}$ and if ${\varepsilon}_1, \ldots\, {\varepsilon}_n$ are the elements of ${\Delta}$ arranged
        according to ${\phi}$ (that is, ${\phi}({\varepsilon}_1) < \ldots < {\phi}({\varepsilon}_n)$).
    \item the empty context otherwise (if ${\Delta} = {\emptyset}$).
  \end{enumerate}~
\end{nota}

\noindent Contexts have the following property~:
\begin{lem}
\label{appendix:mnhnfmhnf1}
  If $C$ is a context and $M$ is a term such that $C[M]$ has a value, then $M$
  has a value.
\end{lem}
\proof
  By case on the form of the context $C$ and by induction on the length of the
  reduction of $C[M]$ to the value. No case raises specific
  difficulties.\qed

\subsection{Orthogonality relation}

\begin{defi}[Orthogonality relation]
  \label{definition:orthogonlatyrelation}
  If $M$ is a term and $C$ a context, then $M\;{\perp}\;C$ (and we say that $M$ and
  $C$ are orthogonal) if and only if $C[M]\rede\daemon$.
\end{defi}

\noindent Moreover, if $S$ is a set of contexts, we define the set of terms $\orth{S}$ by 
\begin{displaymath}
  \orth{S}\quad =\quad \ens{M \md {\forall}\,C\;{\in}\;S,\;M\;{\perp}\;C} 
\end{displaymath}
Note that as with any orthogonality relation, we can easily check that ${\perp}$
verifies the following properties:
\begin{lem}
  \label{lemme:orthoincl}
  If $S$ and $T$ are two context sets such that $S\;{\subseteq}\;T$, then
  $\orth{T}\;{\subseteq}\;\orth{S}$.
\end{lem}
\begin{lem}
  \label{lemme:orthounion}
  If $I$ is any set and $(S_i)_{i\,{\in}\,I}$ is a family of set of contexts
  indexed by $i$, then
  $\orth{(\mathop{{\bigcup}}_{i\,{\in}\,I} S_i)} \ =\  \mathop{{\bigcap}}_{i\,{\in}\,I}\orth{S_i}$.
\end{lem}

\subsection{Operations on sets}

We recall the two standard definition of concatenation $\cons{}{}$ (of a set
of terms and a set of contexts) and composition $\comp{}{}$ (of two sets of
contexts): 
\begin{displaymath}
  \begin{array}{r@{\quad =\quad }l}
    \cons{A}{S} & \ens{C\big[\hole\ N\big] \md C\;{\in}\;S,\ N\;{\in}\;A}\\[1mm]
    \comp{S}{T} & \ens{C[D\hole] \md C\;{\in}\;S,\ D\;{\in}\;T}
  \end{array}
\end{displaymath}
For instance, 
\begin{displaymath}
  \begin{array}{c} 
    \cons{\lens{0,\ 1}}{\lens{\try{\hole}{\daemon}}} \ =\ \lens{\try{(\hole\ 0)}{\daemon},\ \try{(\hole\ 1)}{\daemon}} \\[1mm]
    \comp{\lens{\try{\hole}{\daemon},\ \hole\ (\lam{x}{x})}}{\lens{\hole\ 1}} \ =\ \lens{\try{(\hole\ 1)}{\daemon},\ (\hole\ 1)\ (\lam{x}{x})}
  \end{array}
\end{displaymath}

\noindent We then define two operations on sets of contexts:
\begin{displaymath}
  \begin{array}{r@{\quad =\quad }l}
    \unlift{S} & \comp{S}{\ens{\stry{\hole}}}\\[1mm]
    \lift{S} & \comp{\ens{\stry{\hole}}}{S} 
  \end{array}
\end{displaymath}
and thus, for instance, 
\begin{displaymath}
  \begin{array}{c@{\ {\equiv}\ }c@{\ =\ }l}
    C_1 & \unlift{(\natrec\ 0\ (\lam{x}{\lam{y}{y}})\ \hole)} 
      & \ens{\natrec\ 0\ (\lam{x}{\lam{y}{y}})\ (\stry{\hole})} \\[1mm]
    C_2 & \lift{(\natrec\ 0\ (\lam{x}{\lam{y}{y}})\ \hole)}
      & \ens{\stry{(\natrec\ 0\ (\lam{x}{\lam{y}{y}})\ \hole)}}
  \end{array}
\end{displaymath}
and if ${\varepsilon}\,{\in}\;{\Delta}$, 
\begin{displaymath}
  \begin{array}{l@{\ =\ }l@{\ \rede\ }l}
    C_1[\rai] & \natrec\ 0\ (\lam{x}{\lam{y}{y}})\ (\stry{\rai}) & \daemon \\
    C_1[S\ (\rai)] & \natrec\ 0\ (\lam{x}{\lam{y}{y}})\ (\stry{S\ (\rai)}) & \rai \\[1mm]
    C_2[\rai] & \stry{(\natrec\ 0\ (\lam{x}{\lam{y}{y}})\ (\rai))} & \daemon \\
    C_2[S\ (\rai)] & \stry{(\natrec\ 0\ (\lam{x}{\lam{y}{y}})\ (S\ (\rai)))} & \daemon 
  \end{array}
\end{displaymath}

\noindent It can be checked that by definition we have the following equalities:
\begin{displaymath}
  \begin{array}{r@{\quad =\quad}l}
    \lift{(\cons{A}{S})} & \cons{A}{\lift{S}} \\
    \lift{(\funlift{{\Delta}'}{S})} & \funlift{{\Delta}'}{(\lift{S})}
  \end{array}
\end{displaymath}
Moreover, we have the following lemma:
\begin{lem}
  \label{lemma:model:uniliftortho}
  If $S$ is a set of contexts and ${\Delta}$ and ${\Delta}'$ are sets of exception names, then 
  \begin{displaymath}
    \begin{array}{r@{\quad =\quad}l}
      \orth{(\lift{(\flift{{\Delta}'}{S})})} & \orth{(\flift{{\Delta}{\cup}{\Delta}'}{S})}  \\
      \orth{(\unlift{(\funlift{{\Delta}'}{S})})} & \orth{(\funlift{{\Delta}{\cup}{\Delta}'}{S})}  
    \end{array}
  \end{displaymath}
\end{lem}
\proof
  We only give the proof for $\orth{(\lift{(\flift{{\Delta}'}{S})})} =
  \orth{(\flift{{\Delta}{\cup}{\Delta}'}{S})}$ since the other one is similar. We show two
  inclusions:
  \begin{enumerate}[$\bullet$]
    \item We show that
    $\orth{(\lift{(\flift{{\Delta}'}{S})})}\,{\subseteq}\,\orth{(\flift{{\Delta}{\cup}{\Delta}'}{S})}$:\\
      Let $t\,{\in}\;\orth{(\lift{(\flift{{\Delta}'}{S})})}$ and $C\,{\in}\;\flift{{\Delta}{\cup}{\Delta}'}{S}$,
      by definition $C = \fstry{D}{({\Delta}{\cup}{\Delta}')}$ where $D\,{\in}\;S$. Hence
      $\stry{(\fstry{D[t]}{{\Delta}'})}\rede\daemon$ and
      \begin{enumerate}[$-$]
        \item either $D[t]\rede\daemon$, but then $C[t]\rede\daemon$. 
        \item or $D[t]\rede\rai$ when ${\varepsilon}\,{\in}\;{\Delta}'$ or ${\varepsilon}\,{\in}\;{\Delta}$. But then
        again, $C[t]\rede\daemon$.
      \end{enumerate}
    \item We show
    $\orth{(\flift{{\Delta}{\cup}{\Delta}'}{S})}\,{\subseteq}\,\orth{(\lift{(\flift{{\Delta}'}{S})})}$:\\
      Let $t\,{\in}\;\orth{(\flift{{\Delta}{\cup}{\Delta}'}{S})}$ and $C\,{\in}\;\lift{(\flift{{\Delta}'}{S})}$,
      by definition, $C = \stry{(\fstry{D}{{\Delta}'})}$ where $D\,{\in}\;S$. Hence
      $\fstry{D[t]}{({\Delta}{\cup}{\Delta}')}\rede\daemon$ and
      \begin{enumerate}[$-$]
        \item either $D[t]\rede\daemon$, but then $C[t]\rede\daemon$. 
        \item or $D[t]\rede\rai$ when ${\varepsilon}\,{\in}\;{\Delta}{\cup}{\Delta}'$. But then again,
        $C[t]\rede\daemon$.\qed
      \end{enumerate}
  \end{enumerate}

\noindent Along with the definition of $\lift{}$, this lemma implies
$\orth{(\cons{A}{\lift{(\flift{{\Delta}'}{S})}})}\  =\ \orth{(\cons{A}{\flift{{\Delta}{\cup}{\Delta}'}{S}})}.$

\subsection{Model definition}

We call valuation function any function ${\rho}$ from type variables to the
power set of $\contextset$ minus the empty set (${\rho} :
\arr{\typevarset}{(\mathcal{P}(\contextset))^+}$). To each type $A$ we associate two sets:
\begin{displaymath}
  \begin{array}{l@{\hspace*{1cm}}c@{\;{\subseteq}\;}c@{\hspace*{3cm}}}
    \mbox{A set of contexts} & \strad{A} & \contextset \\
    \mbox{A set of terms} & \trad{A} & \termset
  \end{array}
\end{displaymath}
The set $\trad{A}$ is uniformly defined from $\strad{A}$ by
  $$\trad{A} = \orth{\strad{A}} = \ens{M \md {\forall}C\;{\in}\;\strad{A},\;M\ {\perp}\ C}.$$
The set $\strad{A}$ is defined by induction on $A$. Its definition is given
Figure~\ref{fig:model:modeldef}.
\begin{figurep}
  \begin{framedp}
  \newcommand{\espv}{\vspace*{0.5mm}}
  \begin{displaymath}
    \begin{array}{c@{\quad=\quad}l}
      \strad{{\alpha}} & {\rho}({\alpha}) \espv\\
      \strad{\nat} & \ens{ \natrec\ \daemon\ (\lam{y}{\lam{x}{x}})\ \hole} \espv\\
      \strad{\liste{A}} & \ens{ \foldright\ \daemon\
      (\lam{e}{\lam{l}{\lam{r}{(\strad{A})[e]; r}}})\ \hole} \espv\\
      \strad{\uni{A}} & \unlift{\strad{A}} \espv\\
      \strad{\pw{A}} & \lift{\strad{A}} \espv\\
      \strad{\arr{A}{B}} & \displaystyle \mathop{{\bigcup}}_{{\Delta}\,{\subseteq}\,\exceptionset} 
        \cons{\orth{\strad{\pw{A}}}}{\strad{\pw{B}}} \espv\\
      \strad{\fall{{\alpha}}{A}} & \displaystyle \mathop{{\bigcup}}_{S\,{\subseteq}\,\contextset^+}\fstrad{A}{{\rho};\,{\alpha}\,{\leftarrow}\,S}
    \end{array}
  \end{displaymath}
  \end{framedp}
  \caption{Definition of the $\strad{A}$ set of contexts}
  \label{fig:model:modeldef}
\end{figurep}

Note that the interpretation in the model of the construction $\uni{A}$ and
$\pw{A}$ follows, to some extends, the idea that terms of type $\uni{A}$ are terms
that may raise an exception only at top level, where terms of $\pw{A}$ are
those that may raise an exception in any evaluation context. This is
emphasized by the ``opposition'' of the operations $\unlift$ and $\lift$.
Remark that it is only the restriction to $\daemon$ in the handlers of
$\code{try}$ contexts that allows for such a simple definition of the
interpretation of corruption. Indeed, thanks to this restriction we ensure
that for any context $C$, $C[\rai]$ will always reduce to $\rai$ or $\daemon$.

The other interesting point of the model is the interpretation of arrow types.
In \fx, a function $f$ which has type $\arr{A}{B}$ has also all the types
$\arr{\pw{A}}{\pw{B}}$ for any~${\Delta}$. 
Our arrow type is thus smaller than the usual realizability one and so, functions
of \fx{} are in particular realizability functions. More formally, 
\begin{lem}
  \label{model:arrowiswhatwewant}
  If $A$ and $B$ are two types and ${\rho}$ a valuation function, then
  $$\trad{\arr{A}{B}}\ =\ \displaystyle
  \mathop{\bigcap}_{{\Delta}{\subseteq}\exceptionset}\ens{M\md {\forall}N\,{\in}\;\trad{\pw{A}},\ M\
  N\,{\in}\;\trad{\pw{B}}}.$$
\end{lem}
\proof
  We prove the two inclusions forming the equality separately, both being
  simple consequences of definitions.\qed

We can moreover show that the interpretation of $\uni{A}$ is a union and that
the interpretations of the natural numbers and the lists are
standards:

\begin{lem}
  \label{lemme:fx:model:uniisuni} If $A$ is a type, ${\Delta}$ a set
  of exception names and ${\rho}$ a valuation function, then
  $$\trad{\uni{A}} = \trad{A}\ {\cup}\ \ens{M \md M\rede\rai,\ {\varepsilon}\,{\in}\;{\Delta}}.$$
\end{lem}
\proof
  By definition, 
\[\trad{\uni{A}} = \orth{(\unlift{\strad{A}})} =
  \ens{M\md {\forall}\,C\;{\in}\;\strad{A},\
  C[\demtry{M}]\rede\daemon}\ .
\] 
  We show each side of the inclusion separately:
  \begin{enumerate}[$\bullet$]\overfullrule=2 pt
    \item \overfullrule=2 pt
      For $M\,{\in}\;\ens{t\md {\forall}\,C\;{\in}\;\strad{A},\
      C[\demtry{M}]\rede\daemon}$ we have
      $\demtry{M}\,{\in}\,\trad{A}$. \break Hence $\demtry{M}$ has a value
      (Lemma~\ref{lemme:fx:model:modelwhnf}) and $M$ as well
      (Lemma~\ref{appendix:mnhnfmhnf1})~:
      \begin{enumerate}[$-$]
        \item Either $M\rede\rai$ for ${\varepsilon}\,{\in}\;{\Delta}$ and we conclude directly. 
        \item Or $M\rede V$ where $V$ is a regular value and
          then $\demtry{M}\rede V$. But the interpretation being closed by
          equivalence (Lemma~\ref{lemme:fx:model:cloturebetaequiv}),
          $V\,{\in}\;\trad{A}$ and $M\,{\in}\;\trad{A}$.
      \end{enumerate}
    \item Let $M\,{\in}\;\trad{A}\ {\cup}\ \ens{M \md M\rede\rai,\ {\varepsilon}\,{\in}\;{\Delta}}$ and let
      $C\,{\in}\;\strad{A}$.  We have to show $C[\demtry{M}]\rede\daemon$:
      \begin{enumerate}[$-$]
        \item If $M\rede\rai$ for ${\varepsilon}\,{\in}\;{\Delta}$, then $\demtry{M}\rede\daemon$
          and \m{$C[\demtry{M}]\rede\daemon$}.
        \item Otherwise, $M\,{\in}\;\trad{A}$ and then $\demtry{M}\,{\in}\;\trad{A}$
          (because, $M$ having a value
          (Lemma~\ref{lemme:fx:model:modelwhnf}), there exists  
          $V$ such that $M\rede V$. But then $\demtry{M}\rede V$ and we use
          the closure by equivalence of the interpretation
          (Lemma~\ref{lemme:fx:model:cloturebetaequiv})). Finally, using the
          definition of orthogonality, \m{$C[\demtry{M}]\rede\daemon$}.\qed
      \end{enumerate}
  \end{enumerate}

\begin{lem}
  \label{lemme:natisnat}
  If ${\rho}$ is a valuation function, ${\Delta}$ a set of exception names and if
  ${\Phi}$ represents one of $0$, $\daemon$ or $\rai$ for some ${\varepsilon}\,{\in}\;{\Delta}$,
  then 
  $$\trad{\pw{\nat}} = \ens{M \md M\rede S^n\ {\Phi}, n\,{\in}\;\mathbb{N}}.$$
\end{lem}
\proof
In the following, ${\Phi}$ will always represent one of $0$, $\daemon$ or $\rai$
for some ${\varepsilon}\,{\in}\;{\Delta}$.
\begin{enumerate}[$\bullet$]
  \item If for some integer $n$, $M \rede S^n\ {\Phi}$, then by induction on $n$ it is easy to show that
    $\stry{(\natrec\ \daemon\ (\lam{y}{\lam{x}{x}})\ M)}\rede \daemon$.
  \item If $M\,{\in}\;\trad{\pw{\nat}}$, then there exists $k$ such that
    \m{$\stry{(\natrec\ \daemon\ (\lam{y}{\lam{x}{x}})\ M)}\redk{k} \daemon$}. 
    Hence, we show by induction on $k$ that for any $k'\ {\leq}\ k$ and for any
    term $M$: 
\[\mbox{if}\quad\stry{(\natrec\ \daemon\ (\lam{y}{\lam{x}{x}})\ M)}\redk{k'}
    \daemon\quad\hbox{then}\quad M \rede S^n\ {\Phi}
\]
    \begin{enumerate}[$-$]
      \item We cannot have $k = 0$ since $\stry{(\natrec\ \daemon\
        (\lam{y}{\lam{x}{x}})\ M)}$ is not $\daemon$. 
      \item We have $\stry{(\natrec\ \daemon\ (\lam{y}{\lam{x}{x}})\ M)}\red
        N\redk{k} \daemon$. But the first reduction can only occurs either if
        $M$ is $0$, $\daemon$ or $\rai$ (and in the last case we have
        ${\varepsilon}\,{\in}\;{\Delta}$), or if $M$ is $S\ M'$. In this last case, we easily
        conclude using the reduction of $\natrec$ and the induction
        hypothesis.\qed
    \end{enumerate}
\end{enumerate}

\begin{lem}
  \label{lemme:listislist}
  If ${\rho}$ is a valuation function, $A$ a type, ${\Delta}$ a set of exception names and
  if ${\Phi}$ represents one of $0$, $\daemon$ or $\rai$ for some ${\varepsilon}\,{\in}\;{\Delta}$, then 
  $$\eqalign{&\trad{\pw{(\liste{A})}}\cr ={} &\ens{M \md M\rede 
    \lcons\ a_0\ (\ldots(\lcons\ a_n\ {\Phi})\ldots), n\,{\in}\;\mathbb{N}, 
    {\forall}\,0{\leq}i{\leq}n,\ a_i\,{\in}\;\trad{\pw{A}}}\ .}$$
\end{lem}
\proof
The proof follows the same structure as the one of
Lemma~\ref{lemme:natisnat}.\qed

\subsection{Model properties}

\begin{lem}[closure by equivalence]
  \label{lemme:fx:model:cloturebetaequiv}
  If $M$ and $N$ are two terms, $A$ is a type and ${\rho}$ is a valuation function
  such that $M\,{\in}\;\trad{A}$ and $M\eqred N$, then
  $$N\,{\in}\;\trad{A}.$$
\end{lem}
\proof
  Let $M\,{\in}\;\trad{A}$ and $M\eqred N$. Let $C\,{\in}\;\strad{A}$, by
  definition $C[M]\rede\daemon$. But since $M\eqred N$, $C[M]\eqred C[N]$.
  Thus, by confluence of the reduction (theorem~\ref{syntax:confluence}),
  $C[N]\rede\daemon$ and $N\,{\in}\;\trad{A}$.\qed

In particular, the interpretation is closed by reduction and
anti-reduction. 

\begin{lem}
  \label{lemme:fx:model:modelwhnf}
  If $M$ is a term, $A$ a type and ${\rho}$ a valuation function such that
  $M\,{\in}\;\trad{A}$, then $M$ has a value.
\end{lem}
\proof
  By definition $M\,{\in}\;\trad{A}$ yields $M\,{\in}\;\ens{M\md
  {\forall}C\;{\in}\;\strad{A},\;C[M] \rede \daemon}$.  Thus, if
  $C\,{\in}\;\strad{A}$, $C[M]$ has a value and, using
  Lemma~\ref{appendix:mnhnfmhnf1}, we have that $M$ has a value
  too. We simply have to make sure that there always exists such a
  context $C$, that is $\strad{A}$ is never empty. But for any type
  $A$, it can be easily proved by induction on $A$ that $\strad{A}\
  {\neq}\ {\emptyset}$ and $\daemon\,{\in}\;\trad{A}$ (remark that
  both properties have to be proved simultaneously since the
  non-emptiness of $\strad{A\ {\rightarrow}\ B}$ depends upon the
  non-emptiness of $\trad{A}$ which comes (by induction hypo\-thesis)
  of the non-emptiness of $\strad{A}$).\qed

\begin{lem}
  \label{lemme:fx:model:raiina}
  If $A$ is a type, ${\rho}$ a valuation function and ${\Delta}$ a set of exception names,
  then for all~${\varepsilon}\,{\in}\;{\Delta}$, $\rai\,{\in}\;\trad{\pw{A}}$.
\end{lem}
\proof
  Let $C\,{\in}\;\strad{\pw{A}}$, by definition $C =
  \comp{\lens{\stry{\hole}}}{D}$ where $D\,{\in}\;\strad{A}$. We can easily show
  that either $D[\rai]\rede\daemon$ or $D[\rai]\rede\rai$. In any case,
  $C[\rai]\rede\daemon$.\qed

\noindent An important and essential property of the model is that it validates the
subtyping rule \labelfont{eq-arrc}:
\begin{lem} 
  \label{lemme:fx:model:corruptarrow}
  If $A$ and $B$ are two types, ${\rho}$ is valuation function and~${\Delta}$ is a set of
  exception names, then 
  \begin{displaymath}
    \trad{\pw{(\arr{A}{B})}}\quad =\quad \trad{\arr{\pw{A}}{\pw{B}}}
  \end{displaymath}
\end{lem}
\proof
  Using Lemmas~\ref{lemme:orthounion} and~\ref{lemma:model:uniliftortho}, we
  have that
  \begin{displaymath}
    \begin{array}{c@{\ =\ }c}
      \trad{\pw{(\arr{A}{B})}} & 
        \displaystyle\mathop{\bigcap}_{{\Delta}'\,{\subseteq}\,\exceptionset}\orth{(\cons{\orth{(\flift{{\Delta}'}{\strad{A}})}}{\flift{{\Delta}{\cup}{\Delta}'}{\strad{B}}})}\\
      \trad{\arr{\pw{A}}{\pw{B}}} & 
        \displaystyle\mathop{\bigcap}_{{\Delta}'\,{\subseteq}\,\exceptionset}\orth{(\cons{\orth{(\flift{{\Delta}{\cup}{\Delta}'}{\strad{A}})}}{\flift{{\Delta}{\cup}{\Delta}'}{\strad{B}}})}
    \end{array}
  \end{displaymath}
  It directly follows that 
  \m{$\trad{\pw{(\arr{A}{B})}}\;{\subseteq}\;\trad{\arr{\pw{A}}{\pw{B}}}$}.
   
  For the other inclusion, let $M\,{\in}\;\trad{\arr{\pw{A}}{\pw{B}}}$. If ${\Delta}'$
  is a set of exception names and if 
  $C\,{\in}\;\cons{\orth{(\flift{{\Delta}'}{\strad{A}})}}{\flift{{\Delta}{\cup}{\Delta}'}{\strad{B}}}$, since
  we can show that 
  \m{$\orth{(\flift{{\Delta}'}{\strad{A}})}\;{\subseteq}\;\orth{(\flift{{\Delta}{\cup}{\Delta}'}{\strad{A}})}$}
  we have
  $C\,{\in}\;\cons{\orth{(\flift{{\Delta}{\cup}{\Delta}'}{\strad{A}})}}{\flift{{\Delta}{\cup}{\Delta}'}{\strad{B}}}$ and
  we conclude.\qed 

\begin{lem}
  The interpretations validate the following equalities:
  \begin{displaymath}
    \begin{array}{c@{\quad =\quad }c}
    \strad{\pw{(\fall{{\alpha}}{A})}} & \strad{\fall{{\alpha}}{\pw{A}}} \\ 
    \strad{\uni{(\fall{{\alpha}}{A})}} & \strad{\fall{{\alpha}}{\uni{A}}} \\
    \strad{\fpw{(\uni{A})}{{\Delta}'}} & \strad{{\uni{(\fpw{A}{{\Delta}'})}}} \\
    \trad{\fpw{(\pw{A})}{{\Delta}'}} & \trad{\fpw{A}{{\Delta}{\cup}{\Delta}'}} \\
    \trad{\funi{(\uni{A})}{{\Delta}'}} & \trad{\funi{A}{({\Delta}{\cup}{\Delta}')}}
    \end{array}
  \end{displaymath}
\end{lem}
\proof
  The three first equalities are direct consequence of the definitions, the
  two last are direct consequence of Lemma~\ref{lemma:model:uniliftortho}.\qed

\subsection{Model soundness}

We first show that subtyping is sound with respect to the interpretation we have
defined.
\begin{lem}[Subtyping soudness]
  \label{lemme:model:subtypingsoundness}
  If $A$ and $B$ are two types and ${\rho}$ a valuation function such that 
  $\jst{A}{B},$ then for any set of exception names ${\Delta}$, 
  \begin{displaymath}
    \trad{\pw{A}}\ {\subseteq}\ \trad{\pw{B}}
  \end{displaymath}
\end{lem}
\proof
  We reason by induction on the derivation of~$\jst{A}{B}$. Many cases are
  either trivial (\tlabel{st-id} and \tlabel{st-trans}) or direct consequences
  of the lemmas we have defined so far. We only give in the following the
  cases that do not belong to one of these categories:
\begin{enumerate}[\hbox to8 pt{\hfill}]
\item\noindent{\hskip-12 pt\bf\tlabel{st-arrow}:\ }
      Consider~${\Delta}\,{\subseteq}\;\exceptionset$ and
      \mbox{$M\,{\in}\;\trad{\pw{(\arr{A}{B})}} =
      \trad{\arr{\pw{A}}{\pw{B}}}$}.  We now have to show that 
      \mbox{$M\,{\in}\;\trad{\pw{(\arr{A'}{B'})}}=
      \trad{\arr{\pw{A'}}{\pw{B'}}}$}.  For
      \mbox{$C\,{\in}\;\strad{\arr{\pw{A'}}{\pw{B'}}}$} we will establish
      \mbox{$C[M]\rede\daemon$}. 
      By definition of~\mbox{$\strad{\arr{\pw{A'}}{\pw{B'}}}$}, there
      exist ${\Delta}'$, \mbox{$N\,{\in}\;\trad{\fpw{A'}{{\Delta}{\cup}{\Delta}'}}$}
      and\break \mbox{$D\,{\in}\;\strad{\fpw{B'}{{\Delta}{\cup}{\Delta}'}}$} such that \mbox{$C =
      D[\hole\ N]$}. By induction hypothesis,
      \mbox{$\trad{\fpw{A'}{{\Delta}{\cup}{\Delta}'}}\,{\subseteq}\,\trad{\fpw{A}{{\Delta}{\cup}{\Delta}'}}$}. But
      since~\mbox{$M\,{\in}\;\trad{\arr{\pw{A}}{\pw{B}}}$},
      using Lemma~\ref{model:arrowiswhatwewant}, we have \mbox{$M\
      N\,{\in}\;\trad{\fpw{B}{{\Delta}{\cup}{\Delta}'}}$}. Then by induction hypothesis,
      \mbox{$\trad{\fpw{B}{{\Delta}{\cup}{\Delta}'}}\,{\subseteq}\,\trad{\fpw{B'}{{\Delta}{\cup}{\Delta}'}}$} and finally,
      \mbox{$C[M] = D[M\ N]\rede\daemon$}.
\item\noindent{\hskip-12 pt\bf\tlabel{f-gen}:\ } For ${\Delta}\,{\subseteq}\;\exceptionset$ and
      \m{$M\,{\in}\;\trad{\pw{A}}$} we will show
      that \m{$M\,{\in}\;\trad{\pw{(\fall{{\alpha}}{B})}} = \trad{\fall{{\alpha}}{\pw{B}}}$}.
      Let \m{$C\,{\in}\;\strad{\fall{{\alpha}}{\pw{B}}}$}, there exists~$S$ such that 
      \m{$C\,{\in}\;\fstrad{\pw{B}}{\addtovaluation{{\rho}}{{\alpha}}{S}}$}. Moreover, since
      ${\alpha}\,{\notin}\;\freevarset{A}$, we have \m{$\trad{\pw{A}} =
      \ftrad{\pw{A}}{\addtovaluation{{\rho}}{{\alpha}}{S}}$}.
      And
      since  \m{$\ftrad{\pw{A}}{\addtovaluation{{\rho}}{{\alpha}}{S}}\;{\subseteq}\;\ftrad{\pw{B}}{\addtovaluation{{\rho}}{{\alpha}}{S}}$}
      by induction hypothesis,
      \m{$M\,{\in}\;\ftrad{\pw{B}}{\addtovaluation{{\rho}}{{\alpha}}{S}}$} and finally,
      $C[M]\rede\daemon$.
\item\noindent{\hskip-12 pt\bf\tlabel{f-inst}:\ } Given ${\Delta}\,{\subseteq}\;\exceptionset$, we will show that 
      \mbox{$\strad{\pw{(\subst{A}{{\alpha}}{B})}}\ {\subseteq}\
      \strad{\pw{(\fall{{\alpha}}{A})}}$} and then conclude by
      orthogonality (Lemma~\ref{lemme:orthoincl}). Let
      \m{$C\,{\in}\;\strad{\pw{(\subst{A}{{\alpha}}{B})}}$}, we show
      by a straightforward induction on $A$ that
      \m{$\strad{\pw{(\subst{A}{{\alpha}}{B})}} =
      \fstrad{\pw{A}}{\addtovaluation{{\rho}}{{\alpha}}{\strad{B}}}$}. Moreover, by definition of
      \m{$\strad{\fall{{\alpha}}{\pw{A}}}$}, we have 
      \m{$\fstrad{\pw{A}}{\addtovaluation{{\rho}}{{\alpha}}{\strad{B}}}\ {\subseteq}\
      \strad{\fall{{\alpha}}{\pw{A}}}$}, from which if follows that 
      \m{$C\,{\in}\;\strad{\fall{{\alpha}}{\pw{A}}} = \strad{\pw{(\fall{{\alpha}}{A})}}$}. 
\item\noindent{\hskip-12 pt\bf\tlabel{f-distr}:\ } Consider ${\Delta}\,{\subseteq}\;\exceptionset$ and
      \m{$t\,{\in}\;\trad{\pw{(\fall{{\alpha}}{(\arr{A}{B})})}} =
      \trad{\fall{{\alpha}}{(\arr{\pw{A}}{\pw{B}})}}$}, we will show that
      \m{$t\,{\in}\;\trad{\pw{(\arr{A}{\fall{{\alpha}}{B}})}} =
      \trad{\arr{\pw{A}}{\fall{{\alpha}}{\pw{B}}}}$}. Let
      $C\,{\in}\;\strad{\arr{\pw{A}}{\fall{{\alpha}}{\pw{B}}}}$, by definition there
      exists~${\Delta}'$, \m{$u\,{\in}\;\trad{\fpw{A}{{\Delta}{\cup}{\Delta}'}}$} and
      \m{$D\,{\in}\;\strad{\fpw{(\fall{{\alpha}}{\pw{B}})}{{\Delta}'}}$} \m{$ =
      \strad{\fall{{\alpha}}{\fpw{B}{{\Delta}{\cup}{\Delta}'}}}$} such that
      $C = D[\hole\ u]$. Then there exists~$S$ such
      that~\m{$D\,{\in}\;\fstrad{\fpw{B}{{\Delta}{\cup}{\Delta}'}}{\addtovaluation{{\rho}}{{\alpha}}{S}}$} and
      since~${\alpha}\,{\notin}\;\freevarset{A}$, 
      $u\,{\in}\;\ftrad{\fpw{A}{{\Delta}{\cup}{\Delta}'}}{\addtovaluation{{\rho}}{{\alpha}}{S}}$. Thus by
      definition of \m{$\trad{\fall{{\alpha}}{(\arr{\pw{A}}{\pw{B}})}}$} and using
      Lemma~\ref{model:arrowiswhatwewant}, we have 
      \m{$t\ u\,{\in}\;\ftrad{\fpw{B}{{\Delta}{\cup}{\Delta}'}}{\addtovaluation{{\rho}}{{\alpha}}{S}}$} and
      finally $C[t] = D[t\ u]\rede \daemon$.
\item\noindent{\hskip-12 pt\bf\tlabel{ex-arru}:\ } Consider ${\Delta}\,{\subseteq}\;\exceptionset$ and
      \m{$t\,{\in}\;\trad{\pw{(\funi{(\arr{A}{B})}{{\Delta}'})}}$} \m{$ =
      \trad{\funi{(\arr{\pw{A}}{\pw{B}})}{{\Delta}'}}$}, we have to show that 
      \m{$t\,{\in}\;\trad{\pw{(\arr{A}{\funi{B}{{\Delta}'}})}} =
      \trad{\arr{\pw{A}}{\funi{\pw{B}}{{\Delta}'}}}$}. Consider
      ${\Delta}''\,{\subseteq}\;\exceptionset$ and  $u\,{\in}\;\trad{\fpw{A}{{\Delta}{\cup}{\Delta}''}}$, using
      Lemma~\ref{model:arrowiswhatwewant}, we must show that $t\
      u\,{\in}\;\trad{\funi{\fpw{B}{{\Delta}{\cup}{\Delta}''}}{{\Delta}'}}$. But using
      Lemma~\ref{lemme:fx:model:uniisuni} we have:
      \begin{enumerate}[$\bullet$]
        \item Either $t\rede\rai$ for ${\varepsilon}\,{\in}\;{\Delta}'$, but then $t\ u\rede\rai$
          and hence using Lemma~\ref{lemme:fx:model:uniisuni}, $t\
          u\,{\in}\;\trad{\funi{\fpw{B}{{\Delta}{\cup}{\Delta}''}}{{\Delta}'}}$.
        \item Or $t\,{\in}\;\trad{\arr{\pw{A}}{\pw{B}}}$ and using 
          Lemma~\ref{model:arrowiswhatwewant}, $t\
          u\,{\in}\;\trad{\fpw{B}{{\Delta}{\cup}{\Delta}''}}$ which in turn gives 
          $t\ u\,{\in}\;\trad{\funi{\fpw{B}{{\Delta}{\cup}{\Delta}''}}{{\Delta}'}}$ (Lemma~\ref{lemme:fx:model:uniisuni}).
      \end{enumerate}
\item\noindent{\hskip-12 pt\bf\tlabel{ex-ctx}:\ } This case is trivial with the use of
      Lemma~\ref{lemme:fx:model:uniisuni}. 
\item\noindent{\hskip-12 pt\bf\tlabel{ex-uni}:\ } This case is trivial with the use of
      Lemma~\ref{lemme:fx:model:uniisuni}. 
\item\noindent{\hskip-12 pt\bf\tlabel{ex-corrupt}:\ } We need to show 
      \m{$\trad{\pw{(\funi{A}{{\Delta}'})}}$} \m{$= \trad{\funi{\pw{A}}{{\Delta}'}}$} \m{${\subseteq}\
      \trad{\pw{(\fpw{A}{{\Delta}'})}}= \trad{\fpw{A}{{\Delta}{\cup}{\Delta}'}}$}. But
      \m{$\trad{\funi{\pw{A}}{{\Delta}'}} = \trad{\pw{A}}\;{\cup}\;\ens{t \md
      t\rede\rai,\ {\varepsilon}\,{\in}\;{\Delta}'}$} and it can be easily shown that if
      ${\Delta}\,{\subseteq}\,{\Delta}'$, then $\trad{\pw{A}}\,{\subseteq}\,\trad{\fpw{A}{{\Delta}'}}$
      and that if ${\varepsilon}\,{\in}\;{\Delta}$, then $\rai\,{\in}\;\trad{\pw{A}}$
      (Lemma~\ref{lemme:fx:model:raiina}).\qed
  \end{enumerate}

\noindent We define the interpretation corrupted by some set of exception
names ${\Delta}$ (eventually empty) of a typing context ${\Gamma}$ by: 
\begin{displaymath}
  \trad{\pw{{\Gamma}}} \quad = \quad \ens{{\sigma}\md {\forall}\,(x:A)\;{\in}\;{\Gamma},\ {\sigma}(x)\;{\in}\;\trad{\pw{A}}}
\end{displaymath}
Moreover, if ${\sigma}$ is a substitution of term variables and $M$ is a term, we use the 
notation $M[{\sigma}]$ for the \emph{parallel substitution} of $M$ by ${\sigma}$, which
consists in applying ${\sigma}$ to all free variables of $M$ in parallel. 
We can now show that our interpretation is sound with respect to typing: 
\begin{thm}[Model soundness]
  \label{model:correctness}
  If $M$ is a term, $A$ a type and ${\Gamma}$ a typing context
  such that $\jtyp{M}{A}$, then for all valuation function ${\rho}$, for all set of
  exception names ${\Delta}$ and for all substitution ${\sigma}\;{\in}\;\trad{\pw{{\Gamma}}}$, we have
  $M[{\sigma}]\;{\in}\;\trad{\pw{A}}.$
\end{thm}
\proof
We use induction on the derivation of $\jtyp{t}{A}$. Note that
since \m{$\trad{A}\;{\subseteq}\; \trad{\pw{A}}$}
(Lemma~\ref{lemme:model:subtypingsoundness}), we will only show that
$t[{\sigma}]\,{\in}\;\trad{A}$ when possible. We give here only the interesting
cases. The other cases are either simple (\labelfont{ax}, \labelfont{subs},
\labelfont{zero}, \labelfont{succ}, \labelfont{nil}, \labelfont{cons}) or, for
\labelfont{fold}, follows closely the structure of the proof for
\labelfont{rec}.
\begin{enumerate}[\hbox to8 pt{\hfill}]
\item\noindent{\hskip-12 pt\bf\tlabel{abs}:\ } Let \m{$C\,{\in}\,\strad{\pw{(\arr{A}{B})}}
      = \strad{\arr{\pw{A}}{\pw{B}}}$}, we need to show that
      \m{$C[\lam{x}{t[{\sigma}]}]\rede\daemon$}. By definition, there exists ${\Delta}'$, 
      \m{$u\,{\in}\;\trad{\fpw{A}{{\Delta}'}}$} and
      \m{$D\,{\in}\;\strad{\fpw{B}{({\Delta}{\cup}{\Delta}')}}$} such that \m{$C = D[\hole\ u]$}.
      Then, if \m{${\delta} = {\sigma}+\{x\,{\rightarrow}\,u\}$}, we have 
      \m{${\delta}\,{\in}\;\trad{\fpw{({\Gamma}, x:A)}{({\Delta}{\cup}{\Delta}')}}$} 
      (\m{${\sigma}\,{\in}\;\trad{\fpw{{\Gamma}}{{\Delta}'}}\;{\subseteq}\;\trad{\fpw{{\Gamma}}{({\Delta}{\cup}{\Delta}')}}$} and
      \m{$u\,{\in}\;\trad{\fpw{A}{{\Delta}'}}\;{\subseteq}\trad{\fpw{A}{{\Delta}{\cup}{\Delta}'}}$}),
      and also by induction
      hypothesis, \m{$t[{\delta}]\,{\in}\;\trad{\fpw{B}{({\Delta}{\cup}{\Delta}')}}$}. 
      However, \m{$(\lam{x}{t[{\sigma}]})\ u\red t[{\delta}]$} and $\trad{\fpw{B}{({\Delta}{\cup}{\Delta}')}}$ is
      closed by anti-reduction (Lemma~\ref{lemme:fx:model:cloturebetaequiv}),
      and thus \m{$(\lam{x}{M[{\sigma}]})\ u\,{\in}\;\trad{\fpw{B}{({\Delta}{\cup}{\Delta}')}}$} and finally
      \m{$C[\lam{x}{M[{\sigma}]}] = D[(\lam{x}{M[{\sigma}]})\ u]\rede\daemon$}.
\item\noindent{\hskip-12 pt\bf\tlabel{app}:\ } We easily conclude using Lemma~\ref{model:arrowiswhatwewant}. 
\item\noindent{\hskip-12 pt\bf\tlabel{gen}:\ } Let \m{$C\,{\in}\;\strad{\pw{(\fall{{\alpha}}{A})}} =
      \strad{\fall{{\alpha}}{\pw{A}}}$}, by definition there exists~$S$ non empty
      such that \m{$C\,{\in}\;\fstrad{\pw{A}}{\addtovaluation{{\rho}}{{\alpha}}{S}}$}.
      Moreover, since ${\alpha}\,{\notin}\;\freevarset{{\Gamma}}$, \m{$\trad{\pw{{\Gamma}}} =
      \ftrad{\pw{{\Gamma}}}{\addtovaluation{{\rho}}{{\alpha}}{S}}$}. It follows that by induction
      hypothesis, \m{$t[{\sigma}]\,{\in}\;\ftrad{\pw{A}}{\addtovaluation{{\rho}}{{\alpha}}{S}}$}.
      Finally, \m{$C[t[{\sigma}]]\rede\daemon$} and
      \m{$t[{\sigma}]\,{\in}\;\trad{\pw{(\fall{{\alpha}}{A})}}$}. 
\item\noindent{\hskip-12 pt\bf\tlabel{rec}:\ } We have to show that  
      $$\natrec\,{\in}\;\trad{\fall{{\alpha}}{\arr{\uni{{\alpha}}}{(\arr{\pw{\nat}}{\uni{{\alpha}}}{\uni{{\alpha}}})}{\funi{\pw{\nat}}{{\Delta}'}}{\funi{{\alpha}}{({\Delta}{\cup}{\Delta}')}}}}.$$
      Using Lemma~\ref{model:arrowiswhatwewant}, we have to show that for any
      non empty set of contexts $S$, for any ${\Delta}_1$, ${\Delta}_2$ and ${\Delta}_3$ and for  
      \begin{displaymath}
        \begin{array}{ccl}
          z&{\in}&\ftrad{\uni{\fpw{{\alpha}}{{\Delta}_1}}}{\addtovaluation{{\rho}}{{\alpha}}{S}}\\
          f&{\in}&\ftrad{\arr{\fpw{\nat}{{\Delta}_1{\cup}{\Delta}_2{\cup}{\Delta}}}{\uni{\fpw{{\alpha}}{{\Delta}_1{\cup}{\Delta}_2}}}{\uni{\fpw{{\alpha}}{{\Delta}_1{\cup}{\Delta}_2}}}}{\addtovaluation{{\rho}}{{\alpha}}{S}}\\
          n&{\in}&\ftrad{\funi{\fpw{\nat}{{\Delta}_1{\cup}{\Delta}_2{\cup}{\Delta}_3}}{{\Delta}'}}{\addtovaluation{{\rho}}{{\alpha}}{S}}
        \end{array}
      \end{displaymath}
      we have
      \m{$\natrec\ z\ f\ n\,{\in}\;\ftrad{\funi{\fpw{{\alpha}}{{\Delta}_1{\cup}{\Delta}_2{\cup}{\Delta}_3}}{({\Delta}{\cup}{\Delta}')}}{\addtovaluation{{\rho}}{{\alpha}}{S}}$}. 
      With Lemma~\ref{lemme:fx:model:uniisuni}, we have either $n\rede\rai$ for 
      ${\varepsilon}\,{\in}\;{\Delta}'$ (and we easily conclude), or
      $n\,{\in}\;\ftrad{\fpw{\nat}{{\Delta}_1{\cup}{\Delta}_2{\cup}{\Delta}_3{\cup}{\Delta}}}{\addtovaluation{{\rho}}{{\alpha}}{S}}$. In
      this last case, there exists some $k$ such that $n\rede S^k\ {\Phi}$ where
      ${\Phi}$ is one of $0$, $\daemon$ or $\rai$ for ${\varepsilon}\,{\in}\;{\Delta}_1{\cup}{\Delta}_2{\cup}{\Delta}_3{\cup}{\Delta}$
      (Lemma~\ref{lemme:natisnat}). We then proceed by induction on $k$. If
      $k = 0$ ($n\rede {\Phi}$), we easily conclude in each case of ${\Phi}$. Otherwise,
      we must show that 
      \m{$\natrec\ z\ f\ (S\ (S^k\ {\Phi}))\,{\in}\;\ftrad{\funi{\fpw{{\alpha}}{{\Delta}_1{\cup}{\Delta}_2{\cup}{\Delta}_3}}{({\Delta}{\cup}{\Delta}')}}{\addtovaluation{{\rho}}{{\alpha}}{S}}$}. 
      But $\natrec\ z\ f\ (S\ (S^k\ {\Phi}))\red f\ (S^k\ {\Phi})\ (\natrec\ z\ f\ (S^k\ {\Phi}))$.
      We then conclude using Lemma~\ref{model:arrowiswhatwewant} with
      $f$, the fact that $S^k\ {\Phi}\,{\in}\;\trad{\fpw{\nat}{{\Delta}_1{\cup}{\Delta}_2{\cup}{\Delta}_3{\cup}{\Delta}}}$
      (Lemma~\ref{lemme:natisnat}) and the induction hypothesis.
\item\noindent{\hskip-12 pt\bf\tlabel{raise}:\ } It is a direct using Lemma~\ref{lemme:fx:model:uniisuni}. 
\item\noindent{\hskip-12 pt\bf\tlabel{try}:\ } 
      By induction hypothesis, 
      \m{$t[{\sigma}]\,{\in}\;\trad{\pw{(\uni{A})}} = \trad{\uni{\pw{A}}}$} and 
      \m{$u[{\sigma}]\,{\in}\;\trad{\pw{A}}$}. Using Lemma~\ref{lemme:fx:model:uniisuni},
      we have either that \m{$t[{\sigma}]\,{\in}\;\trad{\pw{A}}$} or that
      \m{$t[{\sigma}]\rede\rai$}. Since terms inhabiting the interpretation have
      values (Lemma~\ref{lemme:fx:model:modelwhnf}), in both case
      we can show that either \m{$\try{t[{\sigma}]}{u[{\sigma}]}$} reduces to $u[{\sigma}]$ or, if
      $t[{\sigma}]\,{\in}\;\trad{\pw{A}}$, it reduces to some $t'$ such that $t[{\sigma}]\rede
      t'$. In both case, we can conclude.\qed
  \end{enumerate}

%
\noindent Note that in this model, we only consider closed terms by
construction. For this very reason, we cannot establish a strong
normalization theorem using this model. But, from the model, we obtain
a form a weak head normalization theorem (let us recall that values
corresponds to weak head normal form):
\begin{thm}[Weak head normalization]
  \label{model:wtypedwhnf}
  If $M$ is a closed term, $A$ a
  type and ${\Gamma}$ a typing context such that $\jtyp{M}{A}$, then $M$ has a value.
\end{thm}
\proof
  This comes directly from the model soundness theorem and
  Lemma~\ref{lemme:fx:model:modelwhnf}.\qed

The model allows us to prove for instance that our typing of exceptions is safe for the
primitive data types of the natural numbers:
\begin{lem}[type safety for natural numbers]
  If $M$ is a term such that \mbox{$\jtypctx{}{M}{\nat}$}, then $M\rede S^n\
  0$ for some $n{\geq}0$.
\end{lem}
\proof
  If \mbox{$\jtypctx{}{M}{\nat}$}, then with Theorem~\ref{model:correctness},
  $M\,{\in}\;\trad{\nat}$. We conclude using Lemma~\ref{lemme:natisnat} and the
  fact that $M$ is well typed and there is no typing rule for $\daemon$.\qed

Hence, if a program is of the type of the natural numbers, we assure
that it will compute a true natural number without producing errors.

%% file: related.tex
\noindent The static detection of uncaught exceptions has been studied in many works,
based on typing or not. For instance, for the OCaml languages, J.C.~Guzmán and
A.~Suárez~\cite{guzman1994ets} have proposed an extension of the type system
where arrows are annotated by the exceptions a function can raise.  Later,
X.~Leroy and F.~Pessaux~\cite{pessaux1999tba} have proposed a similar system
but have added polymorphism over these annotations. Their solution is
efficient and covers all the Ocaml language, including modularity. However,
all these works consider exceptions in call-by-value languages and rely
heavily on the exceptions-as-control-flow paradigm. 

In call-by-name, it is standard to use monads to encode
exceptions~\cite{wadler1985rfl,spivey1990fte}. We have however already
explained in section~\ref{whichexceptions} the drawbacks of such approach. 
As already stressed, from a computational point of view, the exception
mechanism described in this paper is very similar to the imprecise exceptions
of S. Peyton Jones \textit{et al.}~\cite{jones1999sie} who are implemented in
the ghc Haskell compiler~\cite{ghcExceptions}. The novelty of this paper is to
provide a precise type system for this exception mechanism while
in~\cite{jones1999sie} exceptional values inhabit all types. The
``imprecision'' of imprecise exceptions comes from the willingness to not
force a particular reduction strategy for primitive binary operators. For
instance, with imprecise exception the term $(\rai) + (\frai{{\varepsilon}'})$ evaluates
to the set $\lens{\rai, \frai{{\varepsilon}'}}$ (hence exceptional values are sets). Since
in \fx{} we do not have binary primitive operators, we have no need for such so-called
imprecision. However, in~\fx{}, the addition should be coded using the
$\natrec$ operator, such coding being bound to be non commutative for
exceptions (the coding have to choose on which operand of the addition the
recursion should be performed). We however believe that if needed, the typing of
exceptions presented in this paper could be adapted with almost no changes to
the case of imprecise exceptions since our type notions already deal with sets
of exceptions. 


In the literature, exceptions are often considered as control operators. Note
however that exceptions have a dynamic semantic, and as such, cannot be
compared to static control operators like first-class
continuations~\cite{thielecke2002ccc}. In particular, the typing of exceptions
does not necessarily lift the logic to a classical one. Besides, in this paper, we
address the problem of the static detection of uncaught exceptions. We do not
know of previous works on control operators dealing with this particular
problem.

Exceptions in type theoretical settings have been less studied. However, R.~David 
and G.~Mounier~\cite{david2004ilambdac} have designed a typed mechanism 
of exceptions for the language AF2. However, as with monads, the propagation
of exceptions in their system has to be forced by means of Krivine's storage
operators. Besides, their exceptions are restricted in the sense that only
data types can carry exceptions and for example, exceptions cannot be used as
functions.

%% file: futureworks.tex
\noindent We have presented the \fx{} calculus, an extension of
System~F with typed exceptions. We have presented a mechanism of
exceptions that does not force a particular ${\beta}$-reduction
strategy for the calculus. We have also provided a type system for
this mechanism that performs static detection of uncaught
exceptions. This type system is modular and allows the use and
propagation of exceptions to be transparent for the
programmer. Finally, we have justified the semantic of our calculus by
exhibiting a realizability model.

This calculus can be improved in a certain number of ways. First, by proving
more meta-theoretical properties. Our realizability model only allows to prove
weak head normalization but it could probably be modified in order to prove
strong normalization. In fact, we believe that the simple change of the
definition of the orthogonality relation
(definition~\ref{definition:orthogonlatyrelation}) to ``$M\;{\perp}\;C$ if and only if
$C[M]\rede\daemon$ \emph{and} $C[M]$ is strongly normalizing'',  would yield a
strong normalization model (but with this new notion the interpretation
will not be closed by anti-reduction anymore and proofs will have to be
adapted).  Moreover, we have not completed yet the proof of subject-reduction
for \fx.  However, a detailed proof of subject-reduction for the restriction
of the calculus to first-order can be found in~\cite{slphd08} (showing
that corruption does not break intrinsically the subject-reduction property).
Adapting this proof to second order (and thus to \fx) is however not trivial,
not because of corruption, but because of the subtyping rules of
quantification. %
Besides, the realizability model already proves a form a type safety for the
calculus.

Type inference for \fx{} is obviously undecidable~\cite{wells1999tat}. But
type inference for restrictions of \fx{}, to first-order for instance, remains
to be studied, and we have good hopes since we know that in such a
restriction, the subtyping relation is decidable. 

Exceptions in \fx{} are simple names. We would like to extend the calculus so
that they carry arguments. However, we will then need to account in the type
system for the types of these arguments, which will complicates notably the
type system.

As mentioned in the introduction, we think that corruption is a promising
notion for the addition of exceptions to proof assistants based on type
theoretical calculi. To that end, we think that a natural extension would be
to add dependent product to our calculus. As our type system is heavily based
on subtyping, we would build on previous works on subtyping in dependent
calculus~\cite{chen1295scc,miquel2001icc}. Moreover, we already know how to
extend our realizability model to handle the dependent product: if $T$ is a
type and $U_x$ a type family indexed by $x$, we can take 
    \begin{displaymath}
      \strad{{\Pi}x:T.\,U} = \displaystyle
      \mathop{{\bigcup}}_{{\Delta}\,{\subseteq}\,\exceptionset}\ens{\cons{M}{C} \md M\,{\in}\;\trad{\pw{T}}\ {\wedge}\
      C\,{\in}\;\strad{\pw{U_M}}} 
    \end{displaymath}

%% file: appendix.tex
\label{appendix:pred}

\renewcommand{\esp}{\vspace*{0.1cm}}
\begin{displaymath}
  \begin{array}{c@{\qquad}c@{\qquad}c}
    \infer{M\pred M}{} 
    &
    \infer{\lam{x}{M}\pred \lam{x}{M'}}{M\pred M'} 
    &
    \infer{M\ N\pred M'\ N'}{M\pred M' & N\pred N'} 
  \end{array}
\end{displaymath}
\esp
\begin{displaymath}
  \begin{array}{c@{\qquad\qquad}c}
    \infer{(\lam{x}{M})\ N\pred \subst{M'}{x}{N'}}{M\pred M' & N\pred N'} 
    & 
    \infer{(\rai)\ M\pred\rai}{}
  \end{array}
\end{displaymath}
\esp
\begin{displaymath}
  \begin{array}{c@{\qquad}c}
    \infer{\try{(\rai)}{N}\pred N'}{N\pred N'} & 
    \infer{\try{(\frai{{\varepsilon}'})}{N}\pred \frai{{\varepsilon}'}}{} 
  \end{array}
\end{displaymath}
\esp
\begin{displaymath}
  \infer{\try{M}{N}\pred \try{M'}{N'}}{M\pred M' & N\pred N'} 
\end{displaymath}
\esp
\begin{displaymath}
  \infer{\try{V}{N}\pred V'}{V\pred V' & \mbox{$V$ is a regular value}} 
\end{displaymath}
\esp
\begin{displaymath}
  \begin{array}{c@{\qquad}c}
    \infer{\natrec\ X\ Y\ 0\pred X'}{X\pred X'} & 
    \infer{\natrec\ X\ Y\ (S\ N)\pred Y'\ N'\ (\natrec\ X'\ Y'\ N')}{X\pred X' & Y\pred Y' & N\pred N'}
  \end{array}
\end{displaymath}
\esp\esp\esp\esp
\begin{displaymath}
  \infer{\natrec\ X\ Y\ (\rai)\pred \rai}{}
\end{displaymath}
\esp
\begin{displaymath}
  \begin{array}{c@{\quad\enspace}c}
    \infer{\foldright\ X\ Y\ \nil\pred\ X'}{X\pred X'} &
    \infer{\foldright\ X\ Y\ (\lcons\ E\ L)\pred Y'\ E'\ L'\ (\foldright\ X'\ Y'\ L')}{
      X\pred X' & Y\pred Y' & E\pred E' & L\pred L'}
  \end{array}
\end{displaymath}
\esp\esp\esp\esp
\begin{displaymath}
  \infer{\foldright\ X\ Y\ (\rai)\pred \rai}{}
\end{displaymath}